\documentclass{article}

\usepackage[utf8]{inputenc}
\usepackage{adjustbox}
\usepackage[ruled,vlined,noend]{algorithm2e}
\usepackage{amsfonts,amsmath,amssymb,amsthm}
\usepackage{authblk}
\usepackage{bbm}
\usepackage{booktabs}
\usepackage{float}
\usepackage[edges]{forest}
\usepackage[margin=3cm]{geometry}
\usepackage[colorlinks = true,
    linkcolor = blue,
    urlcolor = blue,
    citecolor = teal,
    anchorcolor = blue,
    breaklinks = true]{hyperref}
\usepackage{listings}
\usepackage{mathtools}
\mathtoolsset{showonlyrefs=false}
\usepackage{natbib}
\usepackage{pgfplots}
\usepackage{physics}
\usepackage{standalone}
\usepackage{subcaption}
\usepackage{tikz}
\usepackage{tikz-cd}
\usepackage{xcolor}

\usepackage{easyReview} %

\usepgfplotslibrary{groupplots}
\pgfplotsset{compat=1.14}

\newtheorem{definition}{Definition}[section]

\newtheorem{assumption}{Assumption}[section]

\newtheorem{theorem}{Theorem}[section]

\newtheorem{proposition}{Proposition}[section]

\newtheorem{remark}{Remark}[section]

\newtheorem{lemma}{Lemma}[section]

\newtheorem{corollary}{Corollary}[section]

\SetKwIF{IfReturn}{ElseDo}{Else}{if}{return}{else if}{else}{endif}
\SetKw{Continue}{continue}

\newcommand{\rmref}[0]{\mathrm{ref}}
\newcommand{\lref}[0]{\lambda_{\rmref}}
\newcommand{\tref}[0]{\tau_{\rmref}}

\newcommand{\rmbounce}[0]{\mathrm{bounce}}
\newcommand{\lbounce}[0]{\lambda_{\rmbounce}}
\newcommand{\tbounce}[0]{\tau_{\rmbounce}}
\newcommand{\ppp}[0]{\mathrm{pp}}
\newcommand{\PPP}[0]{\mathrm{PP}}

\title{
    Debiasing Piecewise Deterministic Markov Process samplers using couplings}
\author[1]{
  Adrien Corenflos}
\author[2]{
  Matthew Sutton}
\author[3]{
  Nicolas Chopin}

\affil[1]{Department of Statistics, University of Warwick, \href{adrien.corenflos@warwick.ac.uk}{adrien.corenflos@warwick.ac.uk}}
\affil[2]{The University of Queensland, \href{m.sutton2@uq.edu.au}{matt.sutton@qut.edu.au}}
\affil[3]{ENSAE Paris, Institut Polytechnique de Paris, \href{nicolas.chopin@ensae.fr}{nicolas.chopin@ensae.fr}}

\date{\today}

\begin{document}

    \maketitle
    \begin{abstract}
        Monte Carlo methods --- such as Markov chain Monte Carlo (MCMC) and piecewise deterministic Markov process (PDMP) samplers --- provide asymptotically exact estimators of expectations under a target distribution. There is growing interest in alternatives to this asymptotic regime, in particular in constructing estimators that are exact in the limit of an infinite amount of computing processors, rather than in the limit of an infinite number of Markov iterations. In particular,~\citet{Jacob2020unbiasedMCMC} introduced coupled MCMC estimators to remove the non-asymptotic bias, resulting in MCMC estimators that can be embarrassingly parallelised. In this work, we extend the estimators of~\citet{Jacob2020unbiasedMCMC} to the continuous-time context and derive couplings for the bouncy, the boomerang, and the coordinate samplers. Some preliminary empirical results are included that demonstrate the reasonable scaling of our method with the dimension of the target.
    \end{abstract}

\section{Introduction}\label{sec:introduction}
    
Markov chain Monte Carlo (MCMC) methods are routinely employed to compute
expectations with respect to a probability distribution $\tilde{\pi}(x)$, which
is typically only known up to proportionality, i.e., 
$\tilde{\pi}(x) = \pi(x) / \mathcal{Z}$, where $\mathcal{Z}$ is an intractable 
normalisation constant. At heart, these samplers rely on simulating a Markov chain
$(X_k)_{k=0}^{\infty}$ that keeps the target distribution $\pi$
invariant~\citep[see, e.g.][]{robert2004monte}. Under technical assumptions,
these chains can be shown to verify an ergodic theorem, i.e., $N^{-1}\sum_{k=0}^{N-1}
h(X_k) \rightarrow \pi(h)$  as $N\rightarrow \infty$, almost surely, and
central limit theorems quantifying the deviation to the true mean for a finite
time. In fact, explicit rates of convergence with respect to $N$ are
available~\citep[Ch. 5]{Douc2018MC} and some recent non-asymptotic contraction results have started to emerge~\citep{andrieu2022poincar}. 

MCMC methods have enjoyed a large success since their inception in~\citet{metropolis1953equation} due in part to their wide applicability. A number of improvements on the general method~\citep{Hastings1970mcmc} have subsequently been proposed, the most important ones perhaps being Gibbs samplers~\citep{geman1984stochastic}, that rely on conditional tractability, and gradient-based methods~\citep{duane1987HMC}. However, an issue typically shared by most instances of these methods is their diffusive nature, coming from the fact they verify the detailed balance condition, or, otherwise said, the Markov chain they construct is reversible~\citep[Section 6.5.3]{robert2004monte}. This reversibility implies they explore space at a suboptimal rate (informally, requiring $K^2$ steps to explore $K$ states) compared to their ideal behaviour. Following early works on non-reversibility~\citep{diaconis2000nonrev}, several efforts have been directed at deriving more efficient samplers, relying on the global balance condition only~\citep[Definition 6.3.5]{robert2004monte}, or more precisely, on skewed versions of the detailed balance condition~\citep{chen1999lifting} that work by introducing a momentum to the sampling procedure by means of an auxiliary variable.

In line with this, another class of samplers, called piecewise deterministic Markov processes~\citep[PDMPs,][]{Bouchard2018BPS,bierkens2019zig,durmus2021piecewise} has started to gain attention. 
These samplers can typically be understood as the continuous-time limit of otherwise well-known Markov chains~\citep{peters2012rejection, michel2014generalized, fearnhead2018piecewise}. Formally, they rely on simulating a Markov process $t \mapsto X_t$, ergodic with respect to a target distribution $\pi$. This is usually done by augmenting the space with an auxiliary velocity component $V_t$, such that the augmented target distribution is $\pi(x, v) = \pi(x) \varphi(v)$.
Under similar hypotheses as for the discrete case, ergodic theorems are available: $T^{-1}\int_{0}^{T}
h(X_s)\dd{s} \rightarrow \pi(h)$ as $T \to \infty$. While this formulation may only be useful for integrands $h$ tractable over the PDMP trajectories (typically piecewise polynomial integrands), these samplers can also be discretized to allow for more general test functions $h$, resulting in Markov chains $X_{\Delta k}$, $k=0, 1, \ldots$, for which the standard MCMC theory applies. An attractive characteristic of PDMPs is their non-reversible nature, making them more statistically efficient than their Metropolis--Rosenbluth--Teller--Hastings (MRTH) counterparts~\citep{fearnhead2018piecewise}.

Both MCMC and PDMP samplers provide asymptotically consistent estimates of the integral $\tilde{\pi}(\varphi)$ of interest, in the sense that they are exact
in the regime of an infinite integration time. The picture is
however blurrier in the context of a finite $N$ (or $T$), for which these samplers
exhibit a bias depending on (i) the choice of the initial distribution of
$X_0$, (ii) the construction of the Markov evolution procedure.
Formally, this means that we do not necessarily have
$\mathbb{E}\left[N^{-1}\sum_{k=0}^{N-1} h(X_k)\right] = \pi(h)$ (and similar for the continuous-time system). This issue
hinders parallelisation, indeed, if $\hat{h}_p = N^{-1}\sum_{k=0}^{N-1}
h(X^p_k)$, $p=1, \ldots, P$ are i.i.d. realisations
of the finite sample estimators, their empirical averages $P^{-1} \sum_{p=1}^P \hat{h}_p$ do not converge to $\pi(h)$ and instead capture the finite sample bias.
This issue is typically handled by introducing a burn-in period during which the statistics are not collected: 
$\hat{h}^B_p = (N - B)^{-1}\sum_{k=B}^{N-1} h(X^p_k)$.
The intuition behind this is that, provided that $B$ is large enough, $X_{B}$ will have reached stationarity, i.e., it will be approximately distributed according to $\pi$ so that all subsequent time steps provide almost exact, albeit correlated samples from $\pi$. While sometimes satisfactory in practice, this procedure is still biased, and its bias is hard to quantify. The work of \citet{Jacob2020unbiasedMCMC}, which we review below, is a correction of this procedure for MCMC estimates, and permits computing low-variance bias correction terms in a finite (but random) time. 

The aim of this article is, therefore, to derive unbiased versions of PDMP-based estimators using the framework of~\citet{Jacob2020unbiasedMCMC}, resulting in unbiased estimators. As a consequence, in the remainder of this introduction, we first present the concept of couplings of distributions and Markov chains, which are at the heart of~\citet{Jacob2020unbiasedMCMC}. We then quickly review the framework of~\citet{Jacob2020unbiasedMCMC}; particular attention is paid to explaining why its MRTH and in particular Hamiltonian Monte Carlo (HMC) instances are suboptimal. Finally, we formally introduce PDMP samplers, both to fix notations and to highlight how we may hope that debiasing them will not raise the same issues as MRTH samplers.

\subsection{Couplings in Monte Carlo}\label{subsec:coupling-intro}
    Given two probability distributions $\mu$ and $\nu$ defined on two spaces $\mathcal{X}$ and $\mathcal{Y}$, a coupling $\Gamma$ of $\mu$ and $\nu$ is defined as a probability distribution over $\mathcal{X} \times \mathcal{Y}$ with marginals $\Gamma(\dd{x}, \mathcal{Y}) = \mu(\dd{x})$ and $\Gamma(\mathcal{X}, \dd{y}) = \nu(\dd{y})$. Classical couplings are given by the common random number coupling~\citep{gal1984optimality}, the antithetic variable coupling~\citep{hammersley1956new}, Wasserstein couplings~\citep[see, e.g.][]{villani2009optimal}, and diagonal couplings, in particular maximal couplings~\citep[see, e.g.][]{Thorisson2000Coupling,Lindvall2002lectures}. For instance, the CRN coupling can easily be implemented provided that there exists some common generative uniform random variable $\zeta$ under which we have $X = F^x(\zeta)$ implies $X \sim \mu$ and $Y = F^y(\zeta)$ implies $Y \sim \nu$, in which case the CRN coupling consists in using the same $\zeta$ for simulating both $X$ and $Y$. In this work, we will mostly be interested in diagonal couplings (in particular maximal couplings). 

    Diagonal couplings are defined as couplings which ensure a positive mass on the diagonal: $\mathbb{P}(X = Y) > 0$ as soon as the domains of $X$ and $Y$ have sufficient overlap, and maximal couplings are diagonal couplings that achieve the maximal diagonal mass $\mathbb{P}(X = Y)$. Simulating from diagonal (or maximal) couplings is typically harder than the CRN coupling, although several methods to do so exist~\citep{Thorisson2000Coupling,Bou2020coupling,corenflos2022rejection,Dau2023complexity} with different requirements and properties. Perhaps, the most general algorithm to sample from a maximal coupling is given by Thorisson's algorithm, which only requires that $\mu$ and $\nu$ are defined on the same space $\mathcal{X} = \mathcal{Y}$ and to have a Radon--Nikodym derivative with respect to each other that can be computed. The procedure is given in Algorithm~\ref{alg:coupling-thorisson}.
    \begin{algorithm}[!tbp]
    \caption{Thorisson's maximal coupling}
    \label{alg:coupling-thorisson}
    \DontPrintSemicolon
    Sample $X \sim \mu$ and $V \sim \mathcal{U}([0, 1])$\tcp*[f]{Independently}\;
    \lIfReturn{$V < \nu(X) / \mu(X)$}{
        $(X, X)$
    }\lElse{
        \While{$\mathrm{True}$}{
            Sample $Y \sim \nu$ and $V \sim \mathcal{U}([0, 1])$\tcp*[f]{Independently}\;
            \lIfReturn{$\mu(Y) / \nu(Y) < V$}{
                $(X,Y)$
            }
        }
    }
    \end{algorithm}
    Finally, when the space $\mathcal{X}$ is finite with cardinal $N$, $\mu$ and $\nu$ can be represented by non-negative weights $\mu^n$, $\nu^n$, $n=1, \ldots, N$. It is then easy to form a maximal coupling of $\mu$ and $\nu$ via, for example, Algorithm~\ref{alg:coupling-discrete}. This construction can be found in~\citet{Thorisson2000Coupling,Lindvall2002lectures}.
    \begin{algorithm}[!tbp]
    \caption{Maximal coupling of discrete distributions}
    \label{alg:coupling-discrete}
    \DontPrintSemicolon
    \For{$n=1, \ldots, N$}{
        Set $\alpha_n = \mu_n \wedge \nu_n$
    }
    Compute $\alpha = \sum_{n=1}^N \alpha_n$\;
    \For{$n=1, \ldots, N$}{
        Set $\mu_n = (\mu_n - \alpha_n)/(1-\alpha)$, $\nu_n = (\nu_n - \alpha_n)/(1-\alpha)$, and $\alpha_n = \alpha_n / \alpha$\;
    }
    Sample $U, V \sim \mathcal{U}([0, 1])$\tcp*[f]{Independently}\;
    \uIf{$V < \alpha$}{
        Find the first index $i$ such that $\sum_{n=1}^i \alpha_n> U$\;
        \textbf{return} $(i, i)$
    }\uElse(\tcp*[f]{Sample the residuals with a common uniform}){
        Find the first index $i$ such that $\sum_{n=1}^i \mu_n >U$\;
        Find the first index $j$ such that $\sum_{n=1}^j \nu_n >U$\;
        \textbf{return} $(i, j)$
    }
    \end{algorithm}
    
    Couplings of random variables have traditionally been used as a theoretical technique, either to study properties of the random variables themselves, or more fruitfully, in the study of Markov chains to prove upper bounds on their mixing time~\citep[see, e.g.,][]{Thorisson2000Coupling,Lindvall2002lectures}. However, literature has recently started to emerge where couplings are used as algorithmic procedures rather than pure theoretical tools. They have indeed been used to perform perfect sampling~\citep[known as ``Coupling from the past'']{propp1996exact,huber2016perfect}, for variance reduction~\citep[for example in][]{goodman2009coupling,giles2015mlmc}, or, more recently, as a means to compute unbiased expectation under MCMC simulation~\citep{glynn_rhee_2014,Jacob2020unbiasedMCMC}. We focus on the latter in the next section.
    
\subsection{Unbiased MCMC via couplings}\label{subsec:jacob-intro}
While introducing a burn-in phase to return $\hat{h}^B = (N - B)^{-1}\sum_{k=B}^{N-1} h(X_k)$, rather than $\hat{h}= N^{-1}\sum_{k=0}^{N-1} h(X^p_k)$, empirically reduces the bias incurred by MCMC
estimators, it does not totally remove it. As a consequence, averaging $P$ independent estimators $\hat{h}^B_p \sim \hat{h}^B$, $p=1, \ldots P$
is still ill-defined, in the sense that $P^{-1} \sum_{p=1}^P \hat{h}^B_p \not \to \pi(h)$ for finite $B$ and $N$, as $P$ goes to infinity. In order to correct this,~\citet{Jacob2020unbiasedMCMC}
consider a bias correction term for MCMC samplers, which is defined in terms of
a telescopic sum inspired by the seminal work of~\citet{glynn_rhee_2014,rhee2015unbiased}. For
a large enough burn-in, their bias correction term is shown to be zero with a
high probability, so that the unbiased estimator is almost the same as the
classical one yet can be principally parallelised. This approach has proven
very successful in a number of application and extension papers spanning graph
colouring~\citep{nguyen2022many}, state-space and intractable
models~\citep{middleton2020unbiased,Jacob2020Smoothing}, and high-dimensional
regression models~\citep{biswas2020coupled,biswas2022coupling}. 

Formally, their
estimators are defined in terms of two coupled $\pi$-ergodic chains
$(X_k)_{k=0}^{\infty}$, $(Y_k)_{k=0}^{\infty}$, which are such that 
(a) they have the same marginal distribution 
(for all $k \geq 0$, the law $\mathcal{L}(X_k)$ of $X_k$ is the same as the law
$\mathcal{L}(Y_k)$ of $Y_k$), and 
(b) $X_k$ and $Y_{t-1}$ ``couple'' in finite time (i.e. 
$\tau=\inf\{t: X_{s}=Y_{s-1}\, \forall s\geq t\}$ is finite almost surely). 
Provided $\mathbb{E}[\tau] < \infty$ and under extra assumptions, the following estimate 
may be computed in finite time, is unbiased and has finite variance: 
\begin{equation}
        \label{eq:unbiased-coupling}
        H_{B:N} \coloneqq \frac{1}{N - B}\sum_{k=B}^{N-1} h(X_k) + \sum_{k=B+1}^{\tau - 1} 
        \min\left(1, \frac{k - B}{N - B}\right)\left\{h(X_{k}) - h(Y_{k-1})\right\}.
\end{equation}
(The second sum means zero whenever $\tau-1 < B+1$.)

Constructing such couplings is a subject of active research, and, while some general recipes are available for some classical Markov kernels~\citep{Heng2019couplings,Wang2021coupling}, in the general case, deriving an efficient coupling requires careful examination of the problem at hand.
Furthermore, when a coupling has been derived, it is not trivial to show that it verifies the assumptions of~\citet{Jacob2020unbiasedMCMC}, even when it empirically outperforms theoretically grounded alternatives~\citep[Section 3.2]{biswas2022coupling}.

Additionally, many MCMC methods do not perform satisfactorily without some  calibration (or tuning) of their hyperparameters.
For instance, to get the best out of Hamiltonian Monte Carlo one may need to tune its mass matrix (which defines the geometry of its velocity component), and the step size of  the numerical integrator~\citep[see, e.g.][Section 4.2]{betancourt2017conceptual}. 
Such hyperparameters are typically adapted during the burn-in (or warm-up) phase and then kept constant for the rest of the sampling procedure; however, the correct procedure for combining tuning in MCMC with the coupling method of \citet{Jacob2020unbiasedMCMC} is not obvious to us (see additional discussion in Appendix~\ref{app:tuning}), and there exists, to the best of our knowledge, no literature on the topic. Consequently, most of the positive, if not all, empirical results regarding the method have been obtained for tuning-free MCMC algorithms~\citep[for instance, in][results for the Euclidean metric were obtained using the identity mass matrix]{Heng2019couplings}.
On the other hand, PDMP samplers are either tuning-free algorithms~\citep{bierkens2019zig}, or at least fairly robust to the choice of their hyperparameters~\citep[Supplementary H.3]{Bouchard2018BPS}. This gives an extra incentive to combine the coupling framework of \citet{Jacob2020unbiasedMCMC} with these.

\subsection{Piecewise Deterministic Markov Process (PDMP) samplers}\label{subsec:pdmp-intro}
We now turn to describing PDMP samplers more formally.
PDMP samplers have recently been introduced as a family of non-reversible
continuous-time general-purpose samplers~\citep[see,
e.g.][]{durmus2021piecewise}.  At heart, they are fully defined by a set of
three interacting components: (i) deterministic flows based on 
ordinary differential equations, (ii) random jump events, determining when the
state is modified, (iii) a Markov kernel describing how it is then modified at jump time. 
Formally, for a given process with a state
$Z_t$, this corresponds to the following procedure:
\begin{enumerate}
    \item Jumps happen at random times, following a Poisson process with
        intensity $\lambda_t = \lambda(Z_t)$.
    \item At jump time, the state is refreshed: $Z_{t} \sim Q_t(\cdot \mid
        Z_t)$.
    \item Between jumps, the state evolves according to a flow: $\dot{Z}_t =
        \phi(t, Z_t)$.
\end{enumerate}
This algorithm is fully described by the infinitesimal generator of the PDMP,
which also determines the stationary distribution of $Z$ when it
exists~\citep{davis1984piecewise}.

    This property has made the introduction of PDMP samplers possible, where a PDMP is constructed specifically to target a given distribution of interest.  
    Successful instances of the method include the Zig-Zag
    sampler~\citep[ZZS,][]{bierkens2019zig}, the coordinate
    sampler~\citep[CS,][]{wu2019coordinate}, the bouncy particle
    sampler~\citep[BPS,][]{Bouchard2018BPS}, and the boomerang sampler~\citep[BS][]{bierkens2020boomerang}. In this work, we focus on the latter three methods. Some discussion is provided in the conclusion of the article as to why ZZS is harder to handle than the other samplers. 
    
Interesting features of the methods listed above are that they are all
essentially tuning-free algorithms, and that do not suffer from poor conditioning the way MCMC methods typically do. 
Some general ergodicity results furthermore exist for 
them~\citep{bierkens2029ergodicity,deligiannidis2019exponential,durmus2020geometric} 
where MCMC methods convergence typically are only understood for the simplest kernels and 
targets~\citep[see, however,][for a promising avenue of research]{andrieu2022poincar}.

Of course, these properties do not come at zero cost and PDMP samplers are
often harder to implement than their MCMC counterparts. The main reason for this
is the need to sample from a Poisson process with a position-dependent rate,
which can typically only be done via Poisson thinning~\citep[see, e.g.,][]{Sutton2023ConvexConcave}, akin to rejection sampling, limiting
their application to the class of targets for which a tractable upper bound to
the rate is known or can be computed. Some work to mitigate this issue has recently been
undertaken~\citep[see, e.g.][]{corbella2022automatic,pagani2022nuzz} but typically incurs an unknown
bias in the simulation procedure. As a consequence, we limit ourselves to the
classical thinning procedure in this article and do not consider these biased solutions further.

\subsection{Organisation and contributions}\label{subsec:contribs}

The main goal of this article is to derive coupled estimators for PDMPs
samplers following the method of~\citet{Jacob2020unbiasedMCMC}. In view of
this, the article is organised in the following way. 
\begin{enumerate}
    \item In Section~\ref{sec:unbiased-estimators-for-continuous-time-markov-processes}, we generalise the coupled estimators method to continuous-time processes and provide sufficient conditions for the resulting estimators to verify the hypotheses of~\citet{Jacob2020unbiasedMCMC}.
    \item In Section~\ref{sec:coupling-the-bouncy-particle-sampler}, we describe a coupling for BPS and BS by combining elements from \citet{durmus2020geometric} with time-synchronisation procedures. We further discuss some conditions under which the BPS coupling verifies the consistency hypotheses of Section~\ref{sec:unbiased-estimators-for-continuous-time-markov-processes}.
    \item In Section~\ref{sec:coupling-the-coordinate-sampler}, we describe a coupling for a modified version of the coordinate sampler in terms of successive time-synchronisation of the two coupled processes. 
    \item Empirical validation of the methods introduced in this article is carried out in Section~\ref{sec:experiments}, where we provide preliminary results on the appeal of the method.%
\end{enumerate}  
        
Throughout the article, we will make use of the following notations: 
$\PPP^{\lambda}_u$ is the distribution of duration until the first event occurring after time $u$ of a (possibly
non-homogeneous) Poisson process with intensity function $\lambda\colon t \mapsto \lambda_t$. We write
$\ppp^{\lambda}_u$ for its density. By a slight abuse of notation, we will keep the same notation even when $\lambda$ is a constant function, and equate $\lambda$ with its constant value.

\section{Unbiased estimators for continuous-time Markov processes}\label{sec:unbiased-estimators-for-continuous-time-markov-processes}
In this section, we introduce a new class of unbiased estimators for continuous-time Monte Carlo. This class of estimators is motivated by delayed couplings between PDMPs we introduce in the subsequent sections, but is valid for a more general class of processes. To begin we define the notion of an asynchronous coupling of continuous-time processes.
    \begin{definition}[$\Delta$-coupling]
        \label{def:async-coupling}
        For a given $\Delta>0$, 
        we say that two $\nu$-ergodic continuous-time stochastic processes $(Z^1_t)_{t\geq 0}$ and $(Z^2_t)_{t\geq 0}$ are $\Delta$-coupled when there exist a random time $0 < \kappa < \infty$, such that for all $t \geq 0$,
        \begin{equation}
            \label{eq:async-coupling}
            Z^1_{\kappa + \Delta + t} = Z^2_{\kappa + t}.
        \end{equation}
        We will call $\Delta$ the time delay for the coupling and $\kappa$ its coupling time.
    \end{definition}
    This will allow us to construct unbiased estimators resembling those appearing in~\citet{Jacob2020unbiasedMCMC}. The first class of estimators uses a discretised version of the path and can be computed no matter what the test function is. The second class makes use of the full paths of the process and typically has a lower variance,  but is restricted to when time integrals along a path can be computed (for instance polynomial test functions relating to the moments of the stationary distribution). 
    
    The estimators presented in this section will rely on the following assumptions.
    \begin{assumption}[Marginal equality]
        \label{ass:marginal-equality}
        For all $t \geq 0$, $\mathcal{L}(Z^1_t) = \mathcal{L}(Z^2_t)$.
    \end{assumption}
    \begin{assumption}[$\Delta$-coupling]
        \label{ass:delta-coupling}
        The two processes $Z^1$ and $Z^2$ are $\Delta$-coupled as per Definition~\ref{def:async-coupling}.
    \end{assumption}
    \begin{assumption}
        [Tails of $\kappa$]
        \label{ass:coupling-tails}
        The $\Delta$-coupling time $\kappa$ is such that $\mathbb{P}(\kappa > t) < a \exp(-bt)$ for $a, b > 0$.
    \end{assumption}

    \begin{remark}[Suboptimality of Assumption~\ref{ass:coupling-tails}]
        We note that the requirement for geometric tails for the coupling time was relaxed in \citet{middleton2020unbiased}. However, because the geometric assumption is verified in practice for the bouncy particle sampler, we keep this stronger assumption for clarity of exposition.
    \end{remark}
    
    The following two hypotheses are needed to ensure the well-posedness of our estimators. The first one, more restrictive, will be useful in the context of discretised PDMPs, while the second one is less so, and will be useful for estimators with continuous PDMPs.
    In both cases, $h:\mathcal{X} \rightarrow \mathbb{R}$ is a fixed measurable test function.

    \begin{assumption}
        [Integrability of $h$, discrete version]
        \label{ass:integrability-v2}
        We have
        \begin{equation}
            \label{eq:idiscrete_ntegrated_ergodicity}
            \pi(h) = \lim_{t \to \infty}\mathbb{E} \left[h(Z_t) \right].
        \end{equation}
        Furthermore, there exists $\eta > 0$ and $D < \infty$ such that, for all $t> 0$,
        \begin{equation}
            \mathbb{E} \left[\abs{h(Z_t)}^{2 + \eta}\right] < D.
        \end{equation}
    \end{assumption}
    \begin{assumption}
        [Integrability of $h$, continuous version]
        \label{ass:integrability-v1}
        Let $\Delta > 0$ be fixed, the following holds:
        \begin{equation}
            \label{eq:integrated_ergodicity}
            \pi(h) = \lim_{t \to \infty}\mathbb{E} \left[ \frac{1}{\Delta}
                \int_{t}^{t + \Delta} h(Z_s) \dd{s} \right].
        \end{equation}
        Furthermore, there exists $\eta > 0$ and $D < \infty$ such that, for all $t> 0$,
        \begin{equation}
            \mathbb{E} \left[\abs{\frac{1}{\Delta}
            \int_{t}^{t + \Delta} h(Z_s) \dd{s}}^{2 + \eta}\right] < D.
        \end{equation}
    \end{assumption}

    \begin{proposition}
        Assumption~\ref{ass:integrability-v2} implies Assumption~\ref{ass:integrability-v1}.
    \end{proposition}
    \begin{proof}
        Using Jensen's inequality, if the second part of Assumption~\ref{ass:integrability-v2} is verified for a given pair $\eta, D > 0$,
        \begin{align*}
            \abs{\frac{1}{\Delta}\int_{t}^{t + \Delta} h(Z_s) \dd{s}}^{2 + \eta}
            &\leq \frac{1}{\Delta}\int_{t}^{t + \Delta} \abs{h(Z_s)}^{2 + \eta}\dd{s}
        \end{align*}
        so that, taking the expectation left and right, the second part of Assumption~\ref{ass:integrability-v1} is verified too.
        On the other hand, this also implies that Fubini's theorem applies, so that, for any $t$, $\mathbb{E} \left[ \frac{1}{\Delta} \int_{t}^{t + \Delta} h(Z_s) \dd{s} \right]$ exists, and we have
        \begin{align*}
            \mathbb{E} \left[ \frac{1}{\Delta} \int_{t}^{t + \Delta} h(Z_s) \dd{s} \right]
            &= \frac{1}{\Delta} \int_{t}^{t + \Delta} \mathbb{E} \left[ h(Z_s) \right]\dd{s} =  \frac{1}{\Delta} \int_{0}^{\Delta} \mathbb{E} \left[ h(Z_{t + s}) \right]\dd{s}.
        \end{align*}
        As per Assumption~\ref{ass:integrability-v2}, for all $s \geq 0$, $\lim_{t\to\infty}\mathbb{E} \left[ h(Z_{t + s})\right]$ exists and is equal to $\pi(h)$.
        We can now apply the triangle and Jensen's inequality to establish
        \begin{align*}
            \mathbb{E} \left[ h(Z_{t + s})\right]
            &\leq \mathbb{E} \left[\abs{h(Z_{t + s})}\right] \\
            &\leq \mathbb{E} \left[\abs{h(Z_{t + s})}^{2 + \eta}\right]^{1 / (2 + \eta)} \leq D^{1 / (2 + \eta)}
        \end{align*}
        uniformly in $t$ and $s$. This ensures that we can apply Lebesgue's dominated convergence to conclude that
        \begin{equation*}
            \pi(h) = \lim_{t \to \infty}\mathbb{E} \left[ \frac{1}{\Delta}
                \int_{t}^{t + \Delta} h(Z_s) \dd{s} \right].
        \end{equation*}
    \end{proof}

    All these assumptions directly mirror that of~\citet{Jacob2020unbiasedMCMC}. In Section~\ref{subsec:assump-PBS}, we give sufficient conditions for them to be verified in the case of the bouncy particle sampler. Intuitively, the case of the boomerang sampler should follow from the BPS, but deriving the same conditions for it is out of the scope of this paper.

    \subsection{Time-discretised estimator}\label{subsec:time-discretised}

A first estimator, which resembles closely the one of
\citet{Jacob2020unbiasedMCMC}, is given by the time-discretised version of the
stochastic process and is derived from the remark that, for any given $\Delta$,
$(Z_{k\Delta})_{k \geq 0}$ forms a standard (discrete time) Markov chain. 
Hence, supposing that $Z^1_t$ and $Z^2_t$ are $\Delta$-coupled, 
$(Z^1_{k \Delta})_{k \geq 0}$ and $(Z^2_{k\Delta})_{k \geq 0}$ 
are then coupled in the sense of~\citet[Section 2.1]{Jacob2020unbiasedMCMC}.
We can consequently define the following estimators.
    \begin{definition}
        [Discretised Rhee \& Glynn, DRG]
        \label{def:drg}
        For two $\Delta$-coupled stochastic processes $Z^1_t$, and $Z^2_t$, the discrete Rhee \& Glynn estimator is defined to be
        \begin{equation}
            \label{eq:drg}
            \begin{split}
                \mathrm{DRG}(k, \Delta)
                &\coloneqq  h(Z^1_{k \Delta}) 
                + \sum_{n=k+1}^{\lfloor\frac{\kappa + \Delta}{\Delta}\rfloor}
                \left\{h(Z^1_{n \Delta})- h(Z^2_{(n - 1)\Delta}) \right\}
            \end{split}
        \end{equation}
        for any $k\geq 1$, $\Delta > 0$.
    \end{definition}
    \begin{proposition}
        \label{prop:consistency-drg}
        Under Assumptions~\ref{ass:marginal-equality},~\ref{ass:delta-coupling},~\ref{ass:coupling-tails}, and~\ref{ass:integrability-v2}, estimator \eqref{eq:drg} has mean $\pi(h)$, finite variance and finite expected coupling time $\tau$.
    \end{proposition}
    \begin{proof}
        The first part of this statement directly follows from the fact that $(Z^1_{k\Delta}, Z^2_{k\Delta})$ is a coupled Markov chain which follows the assumptions of~\citet{Jacob2020unbiasedMCMC}. To ensure this, the only assumption which does not immediately mirror~\citet{Jacob2020unbiasedMCMC} is Assumption~\ref{ass:coupling-tails}.
        As a matter of fact, the coupling time of $(Z^1_{k\Delta}, Z^2_{k\Delta})$ is exactly $\lceil \kappa / \Delta \rceil$, so that sub-exponential tails of $\kappa$ are inherited by the coupling of the discretised PDMPs.
    \end{proof}
    We follow~\citet{Jacob2020unbiasedMCMC}, and assuming a $\Delta$-coupling between $Z^1$ and $Z^2$, with coupling time $\kappa$, we can average the estimators $\mathrm{DRG}(l, \Delta)$ over the values $l=k, k+1, \ldots, m$ to obtain another estimator.
    \begin{definition}
        [Averaged Discretised Rhee \& Glynn, ADRG]
        \label{def:adrg}
        For two $\Delta$-coupled stochastic processes $Z^1_t$, and $Z^2_t$, the averaged discrete Rhee \& Glynn estimator is defined to be
        \begin{equation}
            \label{eq:adrg}
            \begin{split}
                \mathrm{ADRG}(k, m, \Delta)
                &\coloneqq \frac{1}{m - k + 1} \sum_{l=k}^{m} \mathrm{DRG}(l, \Delta) \\
                &= \frac{1}{m - k + 1} \sum_{l=k}^m h(Z^1_{l \Delta}) + \sum_{l=k+1}^{\lfloor\frac{\kappa + \Delta}{\Delta}\rfloor}\min\left(1, \frac{l - k}{m - k + 1}\right) \left\{h(Z^1_{l \Delta})- h(Z^2_{(l - 1) \Delta})\right\}
            \end{split}
        \end{equation}
        for any $\Delta > 0$, $m > k \geq 0$.
    \end{definition}
    In turn, under the same hypotheses as Proposition~\ref{prop:consistency-drg}, estimator  \eqref{eq:adrg} verifies the following corollary.
    \begin{corollary}
        \label{cor:consistency-adrg}
        Under Assumptions~\ref{ass:marginal-equality},~\ref{ass:delta-coupling},~\ref{ass:coupling-tails}, and~\ref{ass:integrability-v2}, estimator~\eqref{eq:adrg} has mean $\pi(h)$, finite variance and finite expected coupling time $\tau$.
    \end{corollary}

    While estimators~\eqref{eq:drg} and~\eqref{eq:adrg} directly mirror~\citet{Jacob2020unbiasedMCMC}, we will see later on that we may sometimes only be able to construct successful couplings for \emph{large enough} values of $\Delta$. As a consequence, using these directly would result in an unreasonable waste of intermediate information, and would be unlikely to outperform their reversible, discrete, Markov chain Monte Carlo counterparts.

    \subsection{Doubly discretised estimator}\label{subsec:discrete-bis}

    In view of the potential inefficiency of the estimators of Definitions~\ref{def:drg} and~\ref{def:adrg}, in this section, we introduce a variation on it, which allows for sub-sampling of the trajectories. Let $\delta > 0$ and $M \geq 1$ be an arbitrary float and integer, respectively. We now suppose that the two processes $Z^1$ and $Z^2$ are $M \delta = \Delta$-coupled. Our second family of estimators now relies on the following observation.

    \begin{remark}\label{rem:doubly-discrete-chain}
        $k \mapsto \mathcal{Z}_{k} \coloneqq \left[Z_{(k + 1) \Delta}, Z_{(k + 1) \Delta - \delta}, \ldots, Z_{(k + 1) \Delta - (M-1) \delta}\right]$ is a Markov chain. Furthermore, if the coupling time of $Z^1, Z^2$ is $K$, then the coupling time of $\mathcal{Z}^1, \mathcal{Z}^2$ is at most $K$, and Assumptions~\ref{ass:marginal-equality},~\ref{ass:delta-coupling},~\ref{ass:coupling-tails},~\ref{ass:integrability-v1} hold for $\mathcal{Z}$ as soon as they hold for $Z$.
    \end{remark}

    We can now define our first doubly-discretised estimator.
        \begin{definition}
        [Doubly Discretised Rhee \& Glynn, DDRG]
        \label{def:ddrg}
        For two $\Delta$-coupled stochastic processes $Z^1_t$, and $Z^2_t$, the doubly discrete Rhee \& Glynn estimator is defined to be
        \begin{equation}
            \label{eq:ddrg}
            \begin{split}
                \mathrm{DDRG}(k, \Delta, M)
                &\coloneqq  \frac{1}{M}\sum_{m=0}^{M-1} h(Z^1_{k \Delta - m\delta})  
                + \sum_{n=k+1}^{\lfloor\frac{\kappa + \Delta}{\Delta}\rfloor}
                \left\{\frac{1}{M}\sum_{m=0}^{M-1} h(Z^1_{n \Delta - m\delta})  - \frac{1}{M}\sum_{m=0}^{M-1} h(Z^2_{(n-1) \Delta - m\delta})  \right\}\\
                &=  \frac{1}{M}\sum_{m=0}^{M-1} h(Z^1_{k \Delta - m\delta})  
                + \frac{1}{M}\sum_{n=k+1}^{\lfloor\frac{\kappa + \Delta}{\Delta}\rfloor}\sum_{m=0}^{M-1}
                \left\{h(Z^1_{n \Delta - m\delta}) - h(Z^2_{(n-1) \Delta - m\delta})  \right\}
            \end{split}
        \end{equation}
        for any $k, M\geq 1$, $\delta > 0$, and $\Delta = M \delta$.
    \end{definition}

    Otherwise said, the estimator~\eqref{eq:ddrg} for $Z$ is the same as the estimator~\eqref{eq:drg} for $\mathcal{Z}$ and the test function $\mathcal{h}(z_1, \ldots, z_M) = \frac{1}{M} \sum_{m=1}^M h(z_m)$. As a consequence, Proposition~\ref{prop:consistency-drg} applies to this version too.
    
    Similarly, we may define an averaged version of the estimator~\eqref{eq:ddrg}.

    \begin{definition}
        [Averaged Doubly Discretised Rhee \& Glynn, ADDRG]
        \label{def:addrg}
        For two $\Delta$-coupled stochastic processes $Z^1_t$, and $Z^2_t$, the averaged doubly discrete Rhee \& Glynn estimator is defined to be
        \begin{align}
            \label{eq:addrg}
                \mathrm{ADDRG}(k, m, \Delta, M) \coloneqq
                & \frac{1}{m - k + 1} \sum_{l=k}^{m} \mathrm{DDRG}(l, \delta, M) \\
                = & \frac{1}{M(m - k + 1)} \sum_{l=k}^m \sum_{j=0}^{M-1}h(Z^1_{l \Delta - j \delta}) \\
                &+ \sum_{l=k+1}^{\lfloor\frac{\kappa + \Delta}{\Delta}\rfloor}\sum_{j=0}^{M-1}\min\left(\frac{1}{M}, \frac{l - k}{M(m - k + 1)}\right) \left\{h(Z^1_{l \Delta - j \delta}) - h(Z^2_{(l - 1) \Delta - j \delta})\right\}
        \end{align}
        for any $\delta > 0$, $m > k \geq 1$, $M \geq 1$, and for $\Delta = M \delta$.
    \end{definition}
    Again, Corollary~\ref{cor:consistency-adrg} applies for the estimator as soon as it applies to Definition~\ref{def:addrg}.
    In the next section, we see how this double discretisation can be taken one step further.
    
    \subsection{Time integrated estimator}\label{subsec:time-integrated}
    Our last family of estimators is amenable to test functions that can be integrated along a trajectory and rely on telescopic sums over integrals. It can be seen as a continuous limit of the estimators of Section~\ref{subsec:discrete-bis}.
    Assuming a $\Delta$-coupling between $Z^1$ and $Z^2$, with coupling time $\kappa$, we have, for $l\geq 1$, 
    \begin{align}
        &\mathbb{E}\left[\frac{1}{\Delta}\int_{l \Delta}^{(l+1)\Delta} h(Z^1_s) \dd{s}\right] \\
        &= \mathbb{E}\left[\frac{1}{\Delta}\int_{k \Delta}^{(k+1)\Delta} h(Z^1_s) \dd{s} \right] + \sum_{n=k+1}^{l}\left\{\mathbb{E}\left[\frac{1}{\Delta}\int_{n \Delta}^{(n+1)\Delta} h(Z^1_s) \dd{s} \right] - \mathbb{E}\left[\frac{1}{\Delta}\int_{(n-1) \Delta}^{n\Delta} h(Z^1_s) \dd{s} \right]\right\} \nonumber \\
        &= \mathbb{E}\left[\frac{1}{\Delta}\int_{k \Delta}^{(k+1)\Delta} h(Z^1_s) \dd{s} \right] + \sum_{n=k+1}^{l}\left\{\mathbb{E}\left[\frac{1}{\Delta}\int_{n \Delta}^{(n+1)\Delta} h(Z^1_s) \dd{s} \right] - \mathbb{E}\left[\frac{1}{\Delta}\int_{(n-1) \Delta}^{n\Delta} h(Z^2_s) \dd{s} \right]\right\} \label{eq:invariance} \\
        &= \mathbb{E}\Bigg[\frac{1}{\Delta}\int_{k \Delta}^{(k+1)\Delta} h(Z^1_s) \dd{s} + \sum_{n=k+1}^{l}\frac{1}{\Delta}\left\{\int_{n \Delta}^{(n+1)\Delta} h(Z^1_s) \dd{s} - \int_{(n - 1) \Delta}^{n\Delta} h(Z^2_s) \dd{s}\right\} \Bigg] \nonumber\\
        &= \mathbb{E}\Bigg[\frac{1}{\Delta}\int_{k \Delta}^{(k+1)\Delta} h(Z^1_s) \dd{s} + \sum_{n=k+1}^{l}\frac{1}{\Delta}\int_{n \Delta}^{(n+1)\Delta}\left\{ h(Z^1_s)- h(Z^2_{s - \Delta}) \right\}\dd{s} \Bigg] \nonumber\\
        &= \mathbb{E}\Bigg[\frac{1}{\Delta}\int_{k \Delta}^{(k+1)\Delta} h(Z^1_s) \dd{s} + \sum_{n=k+1}^{l \wedge \lfloor\frac{\kappa + \Delta}{\Delta}\rfloor}\frac{1}{\Delta}\int_{n \Delta}^{(n+1)\Delta} \left\{h(Z^1_s)- h(Z^2_{s - \Delta}) \right\}\dd{s}\Bigg] \label{eq:equality}
    \end{align}
    where step~\eqref{eq:invariance} comes from Assumption~\ref{ass:marginal-equality}, and step \eqref{eq:equality} from Assumption~\ref{ass:delta-coupling}.
    Taking the limit for $l \to \infty$, we obtain our first continuous-time estimator.
    \begin{definition}
        [Continuous Rhee \& Glynn, CRG]
        \label{def:crg}
        For two $\Delta$-coupled stochastic processes $Z^1_t$, and $Z^2_t$, the continuous Rhee \& Glynn estimator is defined to be
        \begin{equation}
            \label{eq:crg}
            \begin{split}
                \mathrm{CRG}(k, \Delta)
                &\coloneqq \frac{1}{\Delta}\int_{k \Delta}^{(k+1)\Delta} h(Z^1_s) \dd{s} + \sum_{n=k+1}^{\lfloor\frac{\kappa + \Delta}{\Delta}\rfloor}\frac{1}{\Delta}\int_{n \Delta}^{(n+1)\Delta} \left\{h(Z^1_s)- h(Z^2_{s - \Delta}) \right\}\dd{s}
            \end{split}
        \end{equation}
        for any $k\geq 1$, $\Delta > 0$.
    \end{definition}
    An important remark, which underlies the proof of Proposition~\ref{prop:consistency-crg} below is the following one.
    \begin{remark}\label{rem:continuous-is-discrete}
        The estimator in Definition~\ref{def:crg} can be seen as the (discrete-time) estimator of~\citet{Jacob2020unbiasedMCMC} for the Markov chain 
        $\left[Z_{(k + 1) \Delta}, \frac{1}{\Delta} \int_{k \Delta}^{(k+1)\Delta} h(Z_s) \dd{s}\right]$ and a choice of test function $g(x, y) = y$.
    \end{remark}
    This remark is the basis for the following proposition.
    \begin{proposition}
        \label{prop:consistency-crg}
        Under Assumptions~\ref{ass:marginal-equality},~\ref{ass:delta-coupling},~\ref{ass:coupling-tails}, and~\ref{ass:integrability-v1}, the estimator \eqref{eq:crg} has mean $\pi(h)$, finite variance and finite expected coupling time $\tau$.
    \end{proposition}
    \begin{proof}
        Remark~\ref{rem:continuous-is-discrete} makes it possible to apply, for example,
    Theorem 1 of~\citet{middleton2020unbiased} directly; see also~\citet{glynn_rhee_2014,rhee2015unbiased,Jacob2020unbiasedMCMC} for similar results.
    Note in particular that our Assumption~\ref{ass:integrability-v2} implies that their Assumption 1 holds for the test function $g(x, y)=y$ applied to the Markov chain 
    $\left[Z_{(k + 1) \Delta}, \frac{1}{\Delta} \int_{k \Delta}^{(k+1)\Delta} h(Z_s) \dd{s}\right]$. 
    \end{proof}

    As before, we can then define time-averaged versions of \eqref{eq:crg}.
    \begin{definition}
        [Averaged Continuous Rhee \& Glynn, ACRG]
        \label{def:arg}
        For two $\Delta$-coupled stochastic processes $Z^1_t$, and $Z^2_t$, the averaged continuous Rhee \& Glynn estimator is defined to be
        \begin{equation}
            \label{eq:acrg}
            \begin{split}
                \mathrm{ACRG}(k, m, \Delta)
                &\coloneqq \frac{1}{m - k + 1} \sum_{l=k}^{m} \mathrm{CRG}(l, \Delta) \\
                &= \frac{1}{m - k + 1} \sum_{l=k}^{m} \frac{1}{\Delta}\int_{l \Delta}^{(l+1)\Delta} h(Z^1_s) \dd{s} \\
                &\quad+ \sum_{l=k+1}^{\lfloor\frac{\kappa + \Delta}{\Delta}\rfloor}\min\left(1, \frac{l - k}{m - k + 1}\right)\frac{1}{\Delta}\int_{l \Delta}^{(l+1)\Delta} \left\{h(Z^1_s)- h(Z^2_{s - \Delta}) \right\}\dd{s}\\
                &= \frac{1}{(m - k + 1) \Delta} \int_{k \Delta}^{(m+1)\Delta} h(Z^1_s) \dd{s} \\
                &\quad+ \sum_{l=k+1}^{\lfloor\frac{\kappa + \Delta}{\Delta}\rfloor}\min\left(1, \frac{l - k}{m - k + 1}\right)\frac{1}{\Delta}\int_{l \Delta}^{(l+1)\Delta} \left\{h(Z^1_s)- h(Z^2_{s - \Delta}) \right\}\dd{s}
            \end{split}
        \end{equation}
        for any $\Delta > 0$, $m > k \geq 0$.
    \end{definition}
    As in~\citet{Jacob2020unbiasedMCMC}, we note that the two terms appearing in \eqref{eq:acrg} correspond to the classical time integrated estimator for ergodic stochastic processes and a bias correction term, respectively.
    Similarly, as for Proposition~\ref{prop:consistency-crg}, we have the following corollary.
    \begin{corollary}
        \label{prop:consistency-acrg}
        Under Assumptions~\ref{ass:marginal-equality},~\ref{ass:delta-coupling},~\ref{ass:coupling-tails}, and~\ref{ass:integrability-v1}, the estimator \eqref{eq:crg} has mean $\pi(h)$, finite variance and finite expected coupling time $\tau$.
    \end{corollary}

    An interesting point to remark on is that, contrary to the estimators of
   ~\citet{Jacob2020unbiasedMCMC}, the fact that the expected coupling time is
    finite does not \emph{directly} imply that the estimators will have finite compute time. This latter point is typically a feature of the non-explosivity of the process at hand and will be the case when only a finite number of computational events can happen in a given interval of time. This will always be verified in practice: for the coupled version of a PDMP to be practical, the PDMP marginal sampling process needs to be as well.

    \section{Coupling samplers with refreshment events: the BPS and the BS}\label{sec:coupling-the-bouncy-particle-sampler}
    We now describe how we can achieve a time-synchronous coupling for the BPS and the BS, to be used in the estimators of Section~\ref{sec:unbiased-estimators-for-continuous-time-markov-processes}.
    We first quickly review the BPS and go on to describe how each of its components can be coupled. Our coupling construction reuses several elements from~\citet{durmus2020geometric} but provably achieves higher coupling probabilities. The method is then extended to BS, which only requires different handling of its deterministic dynamics.

    \subsection{The bouncy particle sampler}\label{sec:the-bouncy-particle-sampler}
    The bouncy particle sampler (BPS) is a PDMP sampler originally introduced in~\citet{peters2012rejection} for sampling from a target distribution $\pi(x) \propto \exp(-U(x))$ over $\mathbb{R}^d$, known only up to a normalising constant and differentiable.
    Formally, BPS targets the augmented distribution $\pi(x, v) \propto \pi(x) \mathcal{N}(v; 0, I)$ over the space $\mathbb{R}^d \times \mathbb{R}^d$, which recovers $\pi$ marginally.
    In order to simulate $\pi(x, v)$-stationary trajectories, BPS stitches affine segments of random length together. Suppose that the state of the sampler at the beginning of a segment is given by $(x(t_i), v(t_i))$, then $(x(t), v(t))$, $t \geq t_i$ follows the dynamics $x(t) = (t - t_i) v(t_i) + x(t_i)$, and $v(t) = v(t_i)$, corresponding to the coupled ordinary differential equation 
    \[\begin{pmatrix}
        \dot{x} \\ \dot{v}
    \end{pmatrix} = \begin{pmatrix}
                        0 & 1\\ 0 & 0
    \end{pmatrix}\begin{pmatrix}
                     x \\ v
    \end{pmatrix}.
    \]
    This is done until a random time $t_{i+1}$ defined as the first event time of a non-homogeneous Poisson process with rate $\lambda_t = \lambda(x(t), v(t)) \coloneqq \langle \nabla U(x(t)), v(t) \rangle_{+}$ (using the notation $z_+:=\max(z, 0)$).
    At time $t_{i+1}$, the velocity is reflected against a hyperplane defined by the gradient of $U$:
    \begin{equation}
        v(t_{i+1}) = R(x(t_{i+1})) v(t_{i}) = v(t_i) - 2 \frac{\langle\nabla U(x(t_{i+1})), v(t_i) \rangle}{\norm{\nabla U(x(t_{i+1}))}^2} \nabla  U(x(t_{i+1})).
    \end{equation}
    
    While this construction keeps $\pi(x, v)$ invariant, it suffers from reducibility problems~\citep{Bouchard2018BPS,deligiannidis2019exponential}, which in turn induces a lack of ergodicity of the process. In order to mitigate this, \emph{refreshment} events, where the velocity $v$ is drawn from $\mathcal{N}(0, I)$, are further added at random times defined by a homogeneous Poisson process with rate $\lref > 0$. All in all, this forms continuous piecewise-affine trajectories for $x(t)$, and c\`adl\`ag piecewise constant trajectories for $v(t)$. %
    Another way to understand the method is via the superposition theorem~\citep[see, e.g.,][Section 2.3.3]{Bouchard2018BPS}: the fastest event time of two Poisson processes with rates $\lambda_1$, $\lambda_2$ corresponds to the event rate of a Poisson process with rate $\lambda_1 + \lambda_2$. The BPS, therefore, proceeds by sampling an event with rate $\lref + \langle\nabla U(x(t)), v(t) \rangle_{+}$  and then randomly selects which one happened first. This formulation can be seen as a special case of the local BPS~\citep{Bouchard2018BPS} and is given in Algorithm~\ref{alg:bps-summary-2}.
    \begin{algorithm}[tbp]
        \caption{The bouncy particle sampler}
        \label{alg:bps-summary-2}
        \DontPrintSemicolon
        Sample $x(0), v(0) \sim \pi_0(x, v)$ \tcp{for an arbitrary initial distribution $\pi_0$.}
        Set $s = 0$\;
        \While{$\mathrm{True}$}{
            Sample $\tau \sim \PPP^{\lref + \lambda_t}_s$\;
            Set $x(t) = x(s) + (t-s)v(s)$, $t \in (s, s + \tau]$ and $v(t) = v(s)$, $t \in [s, s + \tau)$\;
            Set $s = s + \tau$\;
            \tcp{Pick which event happened}
            Sample $V \sim \mathcal{U}([0, 1])$\;
            \lIf{$V < \frac{\lambda_{s}}{\lambda_{s} + \lref}$}{
                set $v(s) = R(x(s)) v(s - \tau)$
            }\lElse{
                sample $v(s) \sim \mathcal{N}(0, I)$
            }
        }
    \end{algorithm}

    Simulating from $\PPP^{\lambda_t}_u$ for arbitrary $\lambda_t$ and $u$ is in general non-trivial, even if this represents a one-dimensional distribution. This is because the density $\ppp^{\lambda_t}_u(t) = \lambda_t\exp\big(-\int_{u}^t \lambda_s \dd{s}\big)\mathbbm{1}(t \geq u)$ is often intractable so that the corresponding quantile function will be too. Instead, when an upper bound $\bar{\lambda}_t \geq \lambda_t$ with tractable first event simulation is known, a thinning procedure is routinely employed to sample from it. We detail this in Algorithm~\ref{alg:thinning-global}, assuming that the upper bound is valid for all $t \geq 0$. When this hypothesis is not verified, an adaptive version~\citep[see, e.g.,][]{Sutton2023ConvexConcave} of Algorithm~\ref{alg:thinning-global} can be used. 
            \begin{algorithm}[!htb]
                \caption{Thinning with a global upper bound}
                \label{alg:thinning-global}
                \DontPrintSemicolon
                Set $\tau = 0$\;
                \While{$\mathrm{True}$}{
                    Sample $\bar{\tau} \sim \PPP^{t \mapsto \bar{\lambda}_{\tau + t}}_{\tau}$ and $U \sim \mathcal{U}([0, 1])$ \tcp*[f]{independently}\;
                    Set $\tau = \tau + \bar{\tau}$\;
                    \lIfReturn{$U < \lambda_{\tau} / \bar{\lambda}_{\tau}$}{$\tau$}
                }
            \end{algorithm}

    \subsection{Coupling the bouncy particle sampler}\label{subsec:synchronising-the-bps}
    We now turn back to the problem of building delayed couplings of the bouncy particle sampler. Recall that by this we mean that, given a time-delay $\Delta$, we want to form $(x^1(t), v^1(t))$ and $(x^2(t), v^2(t))$, both marginally distributed according to the same BPS, such that there exists a random time $\kappa > 0$ with
    \begin{equation}
        x^1(\kappa + \Delta + t) = x^2(\kappa + t), \quad v^1(\kappa + \Delta + t) = v^2(\kappa + t)
    \end{equation}
    for all $t \geq 0$.

    In order to achieve this $\Delta$-coupling, we will consider a time grid $k \Delta$, $k=0, 1, \ldots$, on which all information on the PDMPs at hand is forgotten beside their positions and velocities. In particular, their already computed next event will be forgotten, and we will simply regenerate them via a favourable coupling. 
    The construction we propose closely follows that of~\citet{durmus2020geometric}, who use it to derive mixing rates for BPS, but provably results in lower coupling times than theirs. We will highlight the differences and the reason why these result in lower coupling times for the PDMPs at hand.

    We decompose this problem in two steps: (i) state matching, which is already present in \citet{durmus2020geometric}, and (ii) time synchronisation, which allows us to still try and couple when the method of \citet{durmus2020geometric} fails.

    \subsubsection{The coupling construction of \texorpdfstring{\citet{durmus2020geometric}}{Durmus et al. (2020)}}\label{subsub:state-coupling}
        In this section, we adapt the coupling of \citet{durmus2020geometric} to implement a $\Delta$-coupling of the bouncy particle sampler. This means that, contrary to their construction, we instead consider lagged BPS processes, rather than fully synchronous ones. Furthermore, while they simply perform the operation once to show a positive coupling probability (provided that enough time has passed), we will apply this coupling repeatedly (every $\Delta$ increment of time), leveraging the Markov property of the BPS to ``forget'' future event times values.

        Consider two BPS processes $(x^1(t), v^1(t))$ and $(x^2(t), v^2(t))$, marginally following the same distribution, and define $t_k = k \Delta$ for a given $\Delta > 0$. We will write $\lambda^i_t$, $i=1, 2$, for the corresponding bouncing event rate.
        At time $t_{k+1}$ for $(x^1, v^1)$ and $t_k$ for $(x^2, v^2)$, we implement the following procedure:
        \begin{enumerate}
            \item For both processes, use the same homogeneous Poisson process to sample from the refreshment event times.
            \item Generate independently the bouncing event times $\tbounce^1$ and $\tbounce^2$.
            \item If both processes' next event is a refreshment event, then sample from a coupling of the velocity process to match either their positions (if their positions are not matched yet), or their velocities (if the positions are already matched).
            \item Repeat until you have reached the next $\Delta$ increment.
        \end{enumerate} 
        The coupling procedure is then successful if position and velocity are coupled \emph{within the allocated} time-frame $\Delta$. And if it is not, the same procedure is applied during the next $\Delta$-window $[t_{k+1} + \Delta, t_{k+1} + 2\Delta]$ (or $[t_{k} + \Delta, t_{k} + 2\Delta]$ for the second process). 
        The construction is detailed in Algorithm~\ref{alg:durmus-coupling}.
        \begin{algorithm}[tbp]
        \caption{Coupling of \citet{durmus2020geometric}}
        \label{alg:durmus-coupling}
        \DontPrintSemicolon
        \LinesNumbered
        Set $f_x = f_v = 0$ \tcp{coupling flags, the processes are $\Delta$-coupled when $f_v = f_x = 1$}
        Sample $\tref \sim \mathrm{Exp}(\lref)$, $\tbounce^1 \sim \PPP^{\lambda_t^1}_{t_{k+1}}$ and $\tbounce^2 \sim \PPP^{\lambda_t^2}_{t_{k}}$ independently\; \label{line:poisson-ref-1}
        Set $s^1 = t_{k+1}$, and $s^2 = t_{k}$ \;
        \While{true}{
            Set $\tau^1 = \min(\tref, \tbounce^1)$ and $\tau^2 = \min(\tref, \tbounce^2)$\;
            \lIf{$\max(s^1 + \tau^1 - t_{k+1}, s^2 + \tau^2- t_k) \geq \Delta$}
            {\textbf{break}}
            Set $x^1(t) = x^1(s^1) + (t - s^1)v^1(s^1)$ for $t \in [s^1, s^1 + \tau^1]$,  $v^1(t) = v^1(s^1)$ for $t \in [s^1, s^1 + \tau^1)$\;
            Set $x^2(t) = x^2(s^2) + (t - s^2)v^2(s^2)$ for $t \in [s^2, s^2 + \tau^2]$,  $v^2(t) = v^2(s^2)$ for $t \in [s^2, s^2 + \tau^2)$\;
            Set $s^1 = s^1 + \tau^1$, and $s^2 = s^2 + \tau^2$ \;
            \uIf(\tcp*[f]{The coupling failed}){$\tref > \min(\tbounce^1, \tbounce^2)$}{\label{line:coupling-failed}
                Advance both processes independently up to time $t_{k+1} + \Delta$, $t_k + \Delta$, respectively\;
            }
            \uElse{
                Sample $\tref \sim \mathrm{Exp}(\lref)$ \tcp{Time to next refreshment event}\label{line:poisson-ref-2}
                Sample $u^1, u^2$ from a maximal coupling of $\mathcal{N}(x^1(s^1), (\tref)^2 I)$, $\mathcal{N}(x^2(s^2), (\tref)^2 I)$\; \label{line:coupling-vels}
                \If{The coupling is successful}{
                    Set $f_v = 1 * f_x$ \tcp{Velocities are coupled}
                    Set $f_x = 1$ \tcp{Positions are coupled}
                }
                \Else{
                    Set $f_v = f_x = 0$ \tcp{Velocities and positions are decoupled}
                }
                Set $v^1(s^1) = \left(u^1 - x^1(s^1)\right) / \tref$, $v^2(s^2) = \left(u^2- x^2(s^2)\right) / \tref$\;
                \lIf(\tcp*[f]{Coupling was successful}){$f_v = f_x = 1$}{\textbf{break}}
                Sample $\tbounce^1 \sim \PPP^{\lambda_t^1}_{s^1}$ and $\tbounce^2 \sim \PPP^{\lambda_t^2}_{s^2}$ independently\;
            }
        }
        Continue sampling the two processes up to time $t_{k+1} + \Delta$ and $t_k + \Delta$ respectively, either using a single process if $f_v = f_x = 1$ or independent ones otherwise.
    \end{algorithm}
    In practice, the only remaining implementation choice consists of the coupling for the velocities at line~\ref{line:coupling-vels}. Because this corresponds to coupling two Gaussians with different means, but the same covariance matrix (scaled identity), we can employ the Gaussian coupling of~\citet{Bou2020coupling}, also typically used in the methods of~\citet{Jacob2020unbiasedMCMC} who call it the ``Reflection-maximal'' coupling. This coupling is given in Algorithm~\ref{alg:reflection-maximal}.
    \begin{algorithm}[tbp]
            \caption{Reflection-maximal coupling for $\mathcal{N}(m_1, \sigma^2 I)$ and $\mathcal{N}(m_2, \sigma^2 I)$}
            \label{alg:reflection-maximal}
            \DontPrintSemicolon
            Set $z = (m_1 - m_2) / \sigma$ and $e = z / \norm{z}$\;
            Sample $V \sim \mathcal{N}(0, I)$ and $U \sim \mathcal{U}([0, 1])$ \tcp*[f]{Independently}\;
            \uIf{
                $\mathcal{N}(V; 0, I) \, U < \mathcal{N}(V + z; 0, I)$\;
            }{
                Set $W = V + z$\;
            }
            \uElse{
                Set $W = V - 2 \left\langle e, V \right\rangle e$\;
            }
            Set $x = m_1 + \sigma V$ and $y = m_2 + \sigma W$\;
            \textbf{return} $(x, y)$
        \end{algorithm}
    
    \subsubsection{Improving on \texorpdfstring{\citet{durmus2020geometric}}{Durmus et al. (2020)}}\label{subsub:improving-durmus}
        In practice, the failure condition $\tref > \min(\tbounce^1, \tbounce^2)$ appearing in Algorithm~\ref{alg:durmus-coupling} line~\ref{line:coupling-failed} is too strict. It appears because the probability of having $\tbounce^1 = \tbounce^2$ is null under the independent sampling hypothesis, making it so that, as soon as the processes decouple in time, by which we mean that $\tau^1 \neq \tau^2$, the coupling is doomed to fail. As a consequence, increasing $\Delta$ will not necessarily result in sizeable increases in the resulting coupling probability. This was not a problem in \citet{durmus2020geometric} as they were interested in proving a positive probability for a ``one-shot'' type of coupling, but becomes one in our case given that we wish to repeat the procedure every $\Delta$ time increment until a successful coupling eventually happens. This issue can be corrected at a low cost, by allowing the process to recouple in time when time synchronisation had been lost.
        In particular, rather than coupling the refreshment events by using the same Poisson process with rate $\lref$ at lines~\ref{line:poisson-ref-1} and~\ref{line:poisson-ref-2}, we can instead sample from a coupling of $\mathrm{Exp}(\lref)$ with itself, maximising the probability 
        \begin{equation}\label{eq:sync-ref}
            \mathbb{P}(s^1 + \tref^1 = s^2 + \tref^2 + \Delta),
        \end{equation}
        noting that this corresponds to the same procedure as in Algorithm~\ref{alg:durmus-coupling} when $s^1 = s^2$.
        This corresponds to sampling from a maximal coupling of shifted Exponential distributions, which is easily done in closed-form using the mixture representation of maximal couplings~\citep[Theorem 19.1.6 in][see Appendix~\ref{app:coupling-tractable} for details]{Douc2018MC}. 
        As a consequence, we can ignore the failure condition line~\ref{line:coupling-failed} altogether given that the process is allowed to ``re-synchronise'' even in case of a de-synchronisation. This means that taking a larger $\Delta$ window is less detrimental than it would have been otherwise. Additionally, even during a ``re-synchronisation'' procedure, it is possible to try to couple the positions of the two samplers as in line~\ref{line:coupling-vels} in Algorithm~\ref{alg:durmus-coupling}, albeit, in this case, this will correspond to a coupling of two Gaussians $\mathcal{N}(x^1(\tau^1), (\tref^1)^2 I)$, $\mathcal{N}(x^2(\tau^2), (\tref^2)^2 I)$ with different covariance matrices, for which maximal couplings are not known~\citep{devroye2022total,corenflos2022rejection} in general. Thankfully, this special case is tractable, given that the covariances are proportional to one another, so we can again use the mixture representation of~\citet[][Theorem 19.1.6]{Douc2018MC} to sample from a maximal coupling thereof. 
        
        While this is already an improvement over~\citet{durmus2020geometric}, we can further enhance the coupling procedure by sampling from a coupling of $\tbounce^1$ and $\tbounce^2$, rather than independently. Similar to~\eqref{eq:sync-ref}, we need to target having a positive probability for the event
        \begin{equation}\label{eq:sync-bounce}
            \mathbb{P}(s^1 + \tbounce^1 = s^2 + \tbounce^2 + \Delta).
        \end{equation}
        This is more complicated than for the refreshment events, given the typical intractability of $\PPP^{\lbounce^i}_{s^i}$, $i=1, 2$. However, when these are simulated using the thinning procedure given by Algorithm~\ref{alg:thinning-global} (or an adaptive version), thanks to upper bounds $\PPP^{\bar{\lambda}^1}_{s^1}$ and $\PPP^{\bar{\lambda}^2}_{s^2}$, it is possible to use a modified version of the coupled rejection sampling method of \citet{corenflos2022rejection} to couple the resulting bouncing event times. The method is summarised in Algorithm~\ref{alg:coupled-thinning}.
        \begin{algorithm}[tbp]
                \caption{Coupled thinning}
                \label{alg:coupled-thinning}
                \DontPrintSemicolon
                \LinesNumbered
                \tcp{Suppose given a diagonal coupling $\Gamma$ of $\PPP_u^{\bar{\lambda}^1}$ and $\PPP_u^{\bar{\lambda}^2}$.}
                Set $\tau^1 = \tau^2 = 0$\;
                \While{$\mathrm{True}$}{
                    Sample $\bar{\tau}^1, \bar{\tau}^2 \sim \Gamma$\ and set $\tau^1 = \tau^1 + \bar{\tau}^1$ and $\tau^2 = \tau^2+ \bar{\tau}^2$\;
                    Sample $U \sim \mathcal{U}([0, 1])$ and set $\mathrm{a}_1 = \lambda^1_{\tau^1} / \bar{\lambda}^1_{\tau^1}$ and $\mathrm{a}_2 = \lambda^2_{\tau^2} / \bar{\lambda}^2_{\tau^2}$\;
                    \lIf{$U< \mathrm{a}_1 \wedge \mathrm{a}_2$}
                    {
                        \textbf{return} $(\tau^1, \tau^2)$\label{line:thinning-successful}
                    }\uElseIf{
                        $U < \mathrm{a}_1$ 
                    }{
                        Continue sampling $\tau^2$ via Algorithm~\ref{alg:thinning-global} then \textbf{return} $(\tau^1, \tau^2)$
                    }\uElseIf{
                        $U < \mathrm{a}_2$
                    }{
                        Continue sampling $\tau^1$ via Algorithm~\ref{alg:thinning-global} then \textbf{return} $(\tau^1, \tau^2)$
                    }
                }
            \end{algorithm}
            Thankfully, the upper bounds $\PPP^{\bar{\lambda}^1}_{s^1}$ and $\PPP^{\bar{\lambda}^2}_{s^2}$ will be amenable to coupling as soon as they have a tractable density function, which is the case for, for example, piecewise polynomial bounds~\citep{Sutton2023ConvexConcave}. This can be done via Thorisson's algorithm~\citep{Thorisson2000Coupling} that we presented in Algorithm~\ref{alg:coupling-thorisson}.
            We note however that, contrary to the approach of Appendix~\ref{app:coupling-tractable}, Thorisson's algorithm presents the disadvantage of resulting in random execution times for the simulation of the overarching coupling itself, adding to that of the thinning procedure. For this reason, when they are practical, closed-form solutions should be preferred. Moreover, a drawback of Thorisson's algorithm is the fact that, despite the resulting coupling being maximal, the joint distribution of the coupling when the coupling is not successful (i.e., the distribution of the residual coupling) is independent: in terms of density, $p(x, y \mid x \neq y) = p(x \mid x \neq y) p(y \mid x \neq y)$. This means that the algorithm will behave erratically before we observe a coupling. For these two reasons, we instead use a modified antithetic coupling~\citep{hammersley1956new}, following a construction given in~\citet[Algorithm 14, Section D.2]{Dau2023complexity}. The procedure is given in Algorithm~\ref{alg:coupling-modified-crn}\footnote{We note that Algorithm~\ref{alg:coupling-modified-crn} can in fact be used to transform more or less any coupling into one that gives a positive mass on the diagonal $x=y$, however, for simplicity, we only state it here in terms of the antithetic coupling.}
            and more details on the choice of the antithetic residuals can be found in Appendix~\ref{app:coupling-tractable}. 
            \begin{algorithm}[tbp]
                \caption{Modified antithetic coupling for distributions with densities $p^1$ and $p^2$}
                \label{alg:coupling-modified-crn}
                \DontPrintSemicolon
                Sample $X, Y$ from an antithetic coupling of $p^1$ and $p^2$\; 
                Sample $U, V \sim \mathcal{U}([0, 1])$ and $Z \sim p_1$\tcp*[f]{Independently}\;
                \uIf{$V < p_2(Z) / p_1(Z)$}{
                    \lIf{$U < p_2(X) / p_1(X)$}{set $X = Z$}
                    \lIf{$U < p_1(Y) / p_2(Y)$}{set $Y = Z$}
                }
                \textbf{return} $(X, Y)$
            \end{algorithm}
            Because this modified coupling relies on an underlying antithetic coupling, we can expect that its distribution will be such that, when $X \neq Y$, we will often obtain a contraction behaviour (similar to the reflection of Algorithm~\ref{alg:reflection-maximal}) inherited from the antithetic coupling rather than fully independent samplers, thereby reducing the variance of the algorithm. More details on this can be found in Appendix~\ref{app:coupling-tractable}.
            
        Overall, combining these two improvements results in stochastically smaller coupling times than simply using the method of Section~\ref{subsub:state-coupling} by construction, while at the same time incurring virtually no additional computation compared to it.

    \subsection{Coupling the boomerang sampler}\label{sec:coupling-bs}
        The BS~\citep{bierkens2020boomerang} is very similar to the BPS, in that it shares the same event rates, and only differs via its bouncing and refreshment events and the trajectories between events.
        Formally, the BS targets an augmented distribution with density
        \begin{equation}\label{eq:boomerang-target}
            \pi(x, v) \propto \exp(-U(x) - \frac{1}{2}(x - x^*)^{\top} \Sigma^{-1} (x - x^*) - \frac{1}{2} v^{\top} \Sigma^{-1} v),
        \end{equation}
        or, otherwise said, a distribution with Radon--Nikodym derivative $\exp(-U(x))$ with respect to a reference Gaussian measure $\mathcal{N}(x^*, \Sigma) \otimes \mathcal{N}(0, \Sigma)$, this latter perspective being useful when designing samplers for infinite-dimensional spaces~\citep{dobson2022infinite}.

        Its refreshment event is given by $v(t_{i+1}) \sim \mathcal{N}(0, \Sigma)$,
        while its bouncing event is given by 
        \begin{equation}
        v(t_{i+1}) = R(x(t_{i+1})) v(t_{i}) = v(t_i) - 2 \frac{\langle\nabla U(x(t_{i+1})), v(t_i) \rangle}{\norm{\Sigma^{1/2}\nabla U(x(t_{i+1}))}^2} \nabla \Sigma U(x(t_{i+1})),
        \end{equation}
        and the trajectories between events are given by 
        $\begin{pmatrix}
        \dot{x} \\ \dot{v}
        \end{pmatrix} = \begin{pmatrix}
                            0 & 1\\ -1 & 0
        \end{pmatrix}\begin{pmatrix}
                         x \\ v
        \end{pmatrix} + \begin{pmatrix}
                         0 \\ x^*
        \end{pmatrix}$, which can be integrated in closed form as 
        \begin{equation}
        \begin{split}
            x(t) &= x^* + (x(0) -x^*) \cos(t) + v(0) \sin(t)\\
            v(t) &= -(x(0) -x^*) \sin(t) + v(0) \cos(t).
        \end{split}
        \end{equation}

        Because of this similarity, coupling the BS is essentially done the same way as coupling the BPS, and, to perform state coupling, it suffices to follow the same steps as in Section~\ref{subsub:state-coupling}, whilst adapting them to the trajectories at hand.

        Suppose that the current event at time $t_1$ and $t_2$ is a refreshment event and that the next events at times $t_1 + \tau_1$ and $t_2 + \tau_2$, $\tau_1=\tau_2 = \tref$ are both refreshment events too, then, to couple the positions $X^1_{t_1 + \tau_1}$ and $X^2_{t_2 + \tau_2}$ we want to form a coupling of the $\mathcal{N}(0, \Sigma)$ distributed velocities $V^1$ and $V^2$ maximising the probability of obtaining
        \begin{multline*}
                 x^* + (X^1_{t_1} - x^*) \cos(\tref) + V^1_{t_1} \sin(\tref) = x^* + (X^2 - x^*) \cos(\tref) + V^2_{t_2} \sin(\tref) \\
                \Longleftrightarrow \quad X^1_{t_1} \cos(\tref) + V^1_{t_1} \sin(\tref) = X^2_{t_2}\cos(\tref) + V^2_{t_2} \sin(\tref)
        \end{multline*}
        which corresponds to forming a coupling of the two Gaussian distributions $\mathcal{N}(\cos(\tref) X^1_{t_1}, \sin(\tref)^2\Sigma)$ and $\mathcal{N}(\cos(\tref) X^2_{t_2}, \sin(\tref)^2\Sigma)$ with different means and the same covariance. This is also easily done via Algorithm~\ref{alg:reflection-maximal}.

        The final coupled BS algorithm is then given by the same procedure as for BPS, albeit for different state coupling equations and bouncing events, and similar improvements as in Section~\ref{subsub:improving-durmus} apply here too. 
        
    \subsection{Verifying the assumptions for BPS}\label{subsec:assump-PBS}
        We now focus on verifying that the consistency Propositions~\ref{prop:consistency-drg},~\ref{prop:consistency-crg}, and their corollaries hold for the bouncy particle sampler.
        An important remark, that we use throughout, is the fact that the coupling construction of Section~\ref{sec:coupling-the-bouncy-particle-sampler} is at least as contractive as that of \citet{durmus2020geometric}. This is because, while they consider time synchronisation through the refreshment events only, and ``give up'' if a bouncing event occurs, we instead allow for re-synchronisation of the processes, which results in a higher probability for $\mathbb{P}_{z^1, z^2}\left(\tau^1 = \tau^2\right)$, so that all their inequalities involving this quantity, in particular, are verified in the case of our coupling. As a consequence, all their results still apply to our particular construction.
        
    \begin{proposition}[]
    \label{prop:bps-coupling-time}
         
         Suppose the following conditions on $U(x):=\log\pi(x)$ hold:
         (a) $U$ is twice differentiable; 
         (b) $(x,v)\rightarrow\|v\|\|\nabla U(x)\|$;
         is integrable with respect to the extended target $\bar{\pi}(x,v)=\pi(x)\mathcal{N}(y;0,I)$;
         (c) $U(x)>0$; 
         (d) $\int\exp\left\{ -U(x)/2\right\} dx<\infty$; 
         (e) $\|U(x)\|\rightarrow+\infty$ as $\|x\|\rightarrow+\infty$;
         and (f) there exists $c\in(0,1)$ such that 
         \begin{align*}
             \lim\inf_{\|x\|\rightarrow+\infty}\left\{ \|\nabla U(x)\|/U^{1-c}(x)\right\} &>0,\\ 
             \lim\sup_{\|x\|\rightarrow+\infty}\left\{ \|\nabla U(x)\|/U^{1-c}(x)\right\} &<+\infty\\ \lim\sup_{\|x\|\rightarrow+\infty}\left\{ \|\nabla^{2}U(x)\|/U^{1-2c}(x)\right\} &<+\infty.
         \end{align*} 
    
    Then there exists $\bar{\Delta} > 0$, such that, for any $\Delta > \bar{\Delta}$, the $\Delta$-coupled BPS of Section~\ref{sec:coupling-the-bouncy-particle-sampler} verifies Assumption~\ref{ass:coupling-tails}
            \end{proposition}
    
        The proof of this proposition (which may be found in the Appendix, Section
        \ref{app:proofs}) amounts to using intermediate results of
        \citet{durmus2020geometric} on the geometric ergodicity of the BPS to show
        that the conditions of Proposition 4 in \citet{Jacob2020unbiasedMCMC} hold.
        
        Whilst the conditions on the potential given in Proposition~\ref{prop:bps-coupling-time} may seem complicated, they are essentially saying that it needs to be a perturbation of a homogeneous function. These conditions are typically easier to verify than most conditions arising for discrete Markov chains samplers, for which one needs to explicitly construct drift conditions and study the behaviour of the coupling on small sets~\citep[Proposition 4]{Jacob2020unbiasedMCMC}. A key ingredient for the simplicity of the condition for the BPS is the fact that all compact sets are small~\citep[Lemma 2 and 10 in][respectively]{deligiannidis2019exponential, durmus2020geometric} for this class of potentials.
        
        \begin{remark}
            Proposition~\ref{prop:bps-coupling-time} ensures that our procedure results in well-behaved estimators for a large enough $\Delta$, its minimal possible value depending on the problem at hand. While this may seem restrictive, we have not seen this to appear empirically, and the estimators of Sections~\ref{subsec:discrete-bis} and~\ref{subsec:time-integrated} will still make use of intermediate trajectory information. 
        \end{remark}        

\section{Coupling the coordinate sampler}\label{sec:coupling-the-coordinate-sampler}

We now introduce a coupling of the coordinate sampler (CS). First, the specifics of CS are introduced, after which we describe how it can be coupled by successive applications of time-synchronisation.

        \subsection{The coordinate sampler}\label{subsec:CS}
            Given a target distribution $\pi(x) \propto \exp(-U(x))$ on $\mathbb{R}^d$, the CS introduced in~\citet{wu2019coordinate} considers the augmented target distribution $\pi(x, v) = \pi(x) \varphi(v)$, where $\varphi$ is the uniform distribution on the set $\mathcal{V}_{\textrm{CS}} \coloneqq \{\pm e_i, i=1, \ldots, d\}$ and $e_i$ denotes the $i$'th element of the Euclidean basis of $\mathbb{R}^d$. In order to sample from this augmented distribution, let $Z_t = (X_t, V_t)$ be the stochastic process defined by the following elements:
            \begin{enumerate}
                \item Flow: $\dv{X_t}{t} = V_t$, $\dv{V_t}{t} = 0$,
                \item Rate: $\lambda_t = \lambda(X_t, V_t) = \langle V_t, \nabla U(X_t)\rangle_+ + \lref$, where $\lref > 0$,
                \item Event: at jump time, the new velocity $v$ is sampled with probability $w_{v} \propto \lambda(X_t, -v)$.
            \end{enumerate}
            This corresponds to Algorithm~\ref{alg:coordinate-summary}.
            \begin{algorithm}[htb]
                \caption{The Coordinate Sampler}
                \label{alg:coordinate-summary}
                \DontPrintSemicolon
                Sample $x(0), v(0) \sim \pi_0(x, v)$ \tcp{from an arbitrary initial distribution $\pi_0$.}
                Set $s = 0$\;
                \While{$\mathrm{True}$}{
                    Sample $\tau \sim \PPP^{\lambda_t}_s$\;
                    Set $x(t) = x(s) + (t-s)v(s)$, $t \in [s, s + \tau]$ and $v(t) = v(s)$, $t \in [s, s + \tau)$\;
                    Set $s = s + \tau$\;
                    Sample $v(s) \sim \mathrm{Cat}(w_v)$ with $w_v\propto \lambda(x(s), -v)$\;
                }
            \end{algorithm}
            Provided that $\lref$ is positive, this is guaranteed to be ergodic for a large class of target distributions~\citep[Theorem 4]{wu2019coordinate}.

        \subsection{Coupling the coordinate sampler}\label{subsec:coupling-CS-bad}
            The only sources of stochasticity in the coordinate sampler are given by its jump times and the renewal process of the velocity. In order to form a $\Delta$-coupling of the CS with itself, this limited flexibility restricts us to working on the stochastic time and the velocity space only. 

            Suppose we are given two coordinate samplers processes $(X_{t^1}^1, V_{t^1}^1)$ and $(X_{t^2}^2, V_{t^2}^2)$ for the same target distribution $\pi(x, v)$ and denote by $X^i_{t^i}(j)$ and $V^i_{t^i}(j)$ the value of the $j$'th component of the position and velocity states for the process $i$ at time $t^i$. We also write $\lambda^1_t$ and $\lambda^2_t$ for the corresponding event rates and $\tau^1$, $\tau^2$ for the resulting durations until the next event. 
            Let us ignore for the moment the fact that we want to construct a $\Delta$-coupling and suppose that we only wish to form a position and velocity coupling no matter the lag.
            
            For the sake of exposition, we tackle first the case when $d = 1$ and we suppose that the stochastic times $t^1 = t^2 = t$ are the same. In this case, we either have 
            $V^1_t(1) = V^2_t(1)$ or $V^1_t(1) = -V^2_t(1)$. We can couple the position of the next event by trying to maximise the probability of the event 
            \begin{equation}\label{eq:sync-CS}
                \mathbb{P}(X^1_{t + \tau_1} = X^2_{t + \tau_2}) 
                    = \mathbb{P}(X^1_t + \tau_1 V_t^1 = X^2_t + \tau_2 V_t^2).
            \end{equation}
            Otherwise said, we want to sample from a diagonal coupling of the next events of non-homogeneous Poisson processes with densities 
            \begin{equation}
                p^1(x \mid X^1_t, V_t^1) = \ppp^{\lambda^1_s}_t\left((x - X^1_t) / V_t^1\right), \quad p^2(x \mid X^2_t, V_t^2) = \ppp^{\lambda^2_s}_t\left((x - X^2_t) / V_t^2\right),
            \end{equation}
            noting that $(x - X^i_t) / V_t^i$ is commensurable to the time elapsed.
            This coupling will have a chance to occur in the following cases:
            \begin{enumerate}
                \item $V^1_t = V^2_t=1$ with maximal probability $\int_{X^1_t \vee X^2_t}^{\infty} \ppp^{\lambda^1_s}_t\left(x - X^1_t\right) \wedge \ppp^{\lambda^2_s}_t\left(x - X^2_t\right) \dd{x}$;
                   \item $V^1_t = V^2_t=-1$ with maximal probability $\int^{X^1_t \wedge X^2_t}_{-\infty} \ppp^{\lambda^1_s}_t\left(X^1_t - x\right) \wedge \ppp^{\lambda^2_s}_t\left( X^2_t- x\right) \dd{x}$;
                \item $V^1_t = -1$ and $V^2_t = 1$ when $X^1_t > X^2_t$ with maximal probability $\int_{X^2_t}^{X^1_t} \ppp^{\lambda^1_s}_t\left(X^1_t - x\right) \wedge \ppp^{\lambda^2_s}_t\left(x - X^2_t\right) \dd{x}$
                \item $V^1_t = 1$ and $V^2_t = -1$ when $X^1_t < X^2_t$ with maximal probability $\int_{X^1_t}^{X^2_t} \ppp^{\lambda^1_s}_t\left(x - X^1_t\right) \wedge \ppp^{\lambda^2_s}_t\left(X^2_t - x\right) \dd{x}$
            \end{enumerate}
            
            As a matter of fact, sampling from such couplings can be done via coupled thinning as in Section~\ref{subsub:improving-durmus}. In this case, when the coupling is successful, the two processes are then evolved up to $t + \tau^1$ and $t + \tau^2$ respectively, and share the same position $X^1_{t + \tau^1} = X^2_{t + \tau^2}$ and their velocities need to be updated. We can now try and couple the velocities $V^1_{t + \tau^1}$ and $V^2_{t + \tau^2}$ by sampling from a maximal coupling of categorical distributions with weights $W^1_v \propto \lambda(X^1_{t + \tau^1}, -v)$ and $W^2_v \propto \lambda(X^2_{t + \tau^2}, -v)$ which is easily done in this case (see Algorithm~\ref{alg:coupling-discrete}). In fact, because we only have one component, this will always succeed given that $X^1_{t + \tau^1} = X^2_{t + \tau^2}$, and the two processes will then be $(\tau^1 - \tau^2)$-coupled in the sense of Definition~\ref{def:async-coupling}.

            This construction extends directly to the multidimensional case by considering exactly the same procedure, but component by component: when the two velocities are active for the same component, that is, when $V^1 = \pm V^2 = \pm e_k$, we can try to couple the component $k$. The coupling will then be final when all components have been coupled. We note that this method may result in some components decoupling from time to time, and as a consequence will take more than $d$ successive events to successfully finish.

            It may seem that we can now apply the methods of Section~\ref{sec:unbiased-estimators-for-continuous-time-markov-processes} to this construction. However, an important point to notice is that the value of $\tau^1 - \tau^2$ is \emph{not known in advance}, and could be negative, in which case the estimators~\eqref{eq:drg} and~\eqref{eq:crg} will not be well-defined. We, therefore, need a way to synchronise the two processes for a given $\Delta > 0$. This is the topic of the next section.
        \subsection{Forming \texorpdfstring{$\Delta$}{Δ}-couplings for the coordinate sampler: the on-and-off coordinate sampler}\label{subsec:delta-coupling-CS}
            We now consider the case when $\Delta > 0$ is given and we want to form a $\Delta$-coupling of a multidimensional CS targeting $\pi(x, v)$. In order to do so, we need to introduce a slight generalisation of the CS process where we allow the velocity to be completely switched off. We call this the on-and-off CS (I/O-CS). Formally, we now consider $\pi(x, v) = \pi(x) \varphi(v)$, where $\varphi$ is the uniform distribution on the set $\mathcal{V}_{\textrm{I/O-CS}} \coloneqq \mathcal{V}_{\textrm{CS}} \cup \{0\}$. The rest of the construction of CS is left unchanged. 
            It is easy to see that the proof of Lemma 1 in~\citet{wu2019coordinate} is still valid with this modification and $\pi(x, v)$ is invariant under the modified process. Other guarantees on the ergodicity of the process can be obtained. We list proof of these in Appendix~\ref{app:i-o-cs}.
            We can now apply the coupling method described in the previous
            section to the $\pi$-invariant $I/O-CS$ process, and couple the event times of the two processes $Z^1_{t^1} = (X^1_{t^1}, V^1_{t^1})$ and $Z^2_{t^2} = (X^2_{t^2}, V^2_{t^2})$, now supposed to start at different times $t^1$ and $t^2$, in the following way:
            \begin{enumerate}
                \item When $V^1_{t^1} = 0$ and $V^2_{t^2} = 0$, we couple the next event time $\tau^1$ and $\tau^2$ to maximise $\mathbb{P}(t^1 + \tau^1 + \Delta = t^2 + \tau^2)$. This is easily done because, in this case, the event times follow a homogeneous Poisson process with rate $\lref > 0$. We call this a $\Delta$-synchronisation.
                \item Otherwise, we use the coupling of Section~\ref{subsec:coupling-CS-bad}. We call this state synchronisation.
            \end{enumerate}
            This will result in a $\Delta$-coupling of $Z^1$ and $Z^2$ provided that a $\Delta$-synchronisation event occurs successfully after all the states have been synchronised.            

            \begin{remark}
                An important point to notice is that the construction of the I/O-CS is not equivalent to randomly freezing the coordinate sampler for a given (random time) independent of the process as doing so would have likely resulted in a non-ergodic process. Rather, the amount of time it stays frozen depends on the value of the process at hand.
                If $\nabla U(x)$ is large, then there is a very low probability of the next event being a switch to the null velocity, and when it does switch to it, the next velocity update will happen quickly. As a consequence, the system will only stay frozen when $\nabla U(x)$ is close to being $0$, otherwise said, when the process is exploring a mode, in which case, it will spend $1/(2d+1)$ of the time frozen, in average. 
            \end{remark}

        \subsection{The coupled on-and-off coordinate sampler algorithm}\label{subsec:coupled-cs-final}
            In order to finalise the construction of our I/O-CS $\Delta$-coupling, we only need to describe how the sampling of the velocities at event times happens. Suppose that an event happens at times $t_1$ and $t_2$ for two I/O-CS processes with positions $X^1_{t_1}$ and $X^2_{t_2}$, respectively. The marginal distributions of the refreshed velocities are given by categorical distributions on $\mathcal{V}_{\textrm{I/O-CS}}$ with weights $W_v^1 \propto \lambda(X^1_{t_1}, -v)$ and $W_v^2 \propto \lambda(X^2_{t_2}, -v)$, respectively. Fortunately, maximal couplings on categorical distributions are well understood~\citep[see, e.g.][and references within]{Thorisson2000Coupling,Lindvall2002lectures}, and can be sampled from easily (see Section~\ref{subsec:coupling-intro}). 

            A natural way to couple the velocity refreshment is then to try and maximise the probability of picking exactly the same velocity component. This would indeed result in a valid method with heuristically good coupling probabilities. 
            However, this would be missing that coupling is often possible when $v^1 = -v^2$ and it is, therefore, more beneficial to try to maximise the probability that the two vectors $v^1$ and $v^2$ are proportional. As a consequence, writing $W^i_{\abs{v}} = W^i_{v} + W^i_{-v}$ for $v \neq 0$ and $W^i_{\abs{0}} = W^i_{0}$, we wish to rather sample from a maximum coupling of $W^1_{\abs{v}}$ and $W^2_{\abs{v}}$. This corresponds to a coupling on random indices $j_1, j_2$, taking values in $\{0, 1, \ldots, d\}$, and with categorical distributions having weights $ W^i_{\abs{v}}$, $j=0, 1, \ldots, N$, $i=1, 2$, respectively.

            Suppose now that the coupling has been successful and that $j_1 = j_2 > 0$, so that we will have $v^1$ and $v^2$ co-linear and non-null. Conditionally on $j_1 = j_2 = j$, the distributions of $v^1$ and $v^2$ are given as categorical distributions on the set $\{\pm e_j\}$ with weights 
            \[
(W^1_+, W^1_-) = \left(\frac{W^1_{e_j}}{W^1_{j}}, \frac{W^1_{-e_j}}{W^1_{j}}\right),
\quad
(W^2_+, W^2_-) = \left(\frac{W^2_{e_j}}{W^2_{j}}, \frac{W^2_{-e_j}}{W^2_{j}}\right).      
            \]

            Now, suppose $X^1_{t_1}(j) > X^2_{t_2}(j)$, for the next event to be able to result in a position coupling, we need either $v^1 = v^2$, or $v^1 = -e_j$ and $v^2 = e_j$. As a consequence, we aim at minimising the probability of the event $\{v^1 = e_j \}\cap\{v^2 = -e_j\}$. 
            This corresponds to using a common random coupling on the two Bernoulli distributions $p(v^1 = v^2 = e_j) = W_{++}$, $p(v^1 = v^2 = -e_j) = W_{--}$, $p(v^1 = -v^2 = e_j) = W_{+-}$, and $p(v^1 = -v^2 = -e_j) = W_{-+}$.
            The same is true for $X^1_{t_1}(j) < X^2_{t_2}(j)$. The overall coupled I/O-CS algorithm is given in Algorithm~\ref{alg:coupled-IO-CS}.
            \begin{algorithm}[!htb]
                \caption{$\Delta$-coupling of the I/O-CS}
                \label{alg:coupled-IO-CS}
                \DontPrintSemicolon
                Sample $x^1(0), v^1(0)$ and $x^2(0), v^2(0)$ independently of $\pi_0(x, v)$.
                Set $s^1 = s^2 = 0$\;
                \While{$\mathrm{True}$}{
                    Sample $\tau^1, \tau^2$ from a coupling of $\PPP^{\lambda^1_t}_{s^1}$ and $\PPP^{\lambda^2_t}_{s^2}$ following Section~\ref{subsec:delta-coupling-CS}\;
                    Set $x^j(t) = x^j(s^j) + (t-s)v^j(s^j)$, $t \in [s^j, s^j + \tau^j]$ and $v^j(t) = v^j(s^j)$, $t \in [s^j, s^j + \tau^j)$, $j=1, 2$\;
                    Set $s^j = s^j + \tau^j$, $j=1, 2$\;
                    Sample $v(s^1), v(s^2)$ from a common random number coupling of $\mathrm{Cat}(W^1)$ and $\mathrm{Cat}(W^2)$ following Section~\ref{subsec:coupled-cs-final}\;
                }
            \end{algorithm}

    \pagebreak

    \section{Experimental validation}\label{sec:experiments}
        In this section, we provide validation and comparison of our coupled samplers' behaviour. We investigate the dependence on tuning parameters, scaling with dimension and efficiency of the coupled estimators.

        \subsection{Empirical scaling of coupling times: Gaussian targets}        
        We first investigate the performance of our different coupled samplers for a simple $d$-dimensional normal distribution $\mathcal{N}(0_d,I_d)$ as we scale up the dimension. For this example we keep the refreshment rate $\lref = 1$ constant for all samplers; a more detailed examination of the tuning of $\Delta$ may be found in Appendix \ref{app:tuneDelta}. We investigate the average $\Delta$-coupling time $\kappa$ defined in \ref{def:async-coupling} as the time for which $Z_{\kappa+\Delta+t}^1=Z_{\kappa+t}^2$ for all $t>0$. We also measure the average computation required for the coupling to occur. The \textit{computational cost} is defined as the sum of the number of gradient evaluations used by the coupled processes. For the coordinate sampler, a single partial derivative rather than full gradient vector is needed so we scale the computational cost by the inverse of the dimension. The coupled Boomerang sampler for this example is implemented with the reference distribution equal to the target. This represents an ideal scenario where the Boomerang follows the Hamiltonian flow and the only events that occur are refreshment events. Consequently no gradients are required to implement the Boomerang and the computation reported is instead the sum of the number of refreshment events used by the coupled processes.
        
        \begin{figure}[!htp]
            \centering
            \includegraphics[width=.8\textwidth]{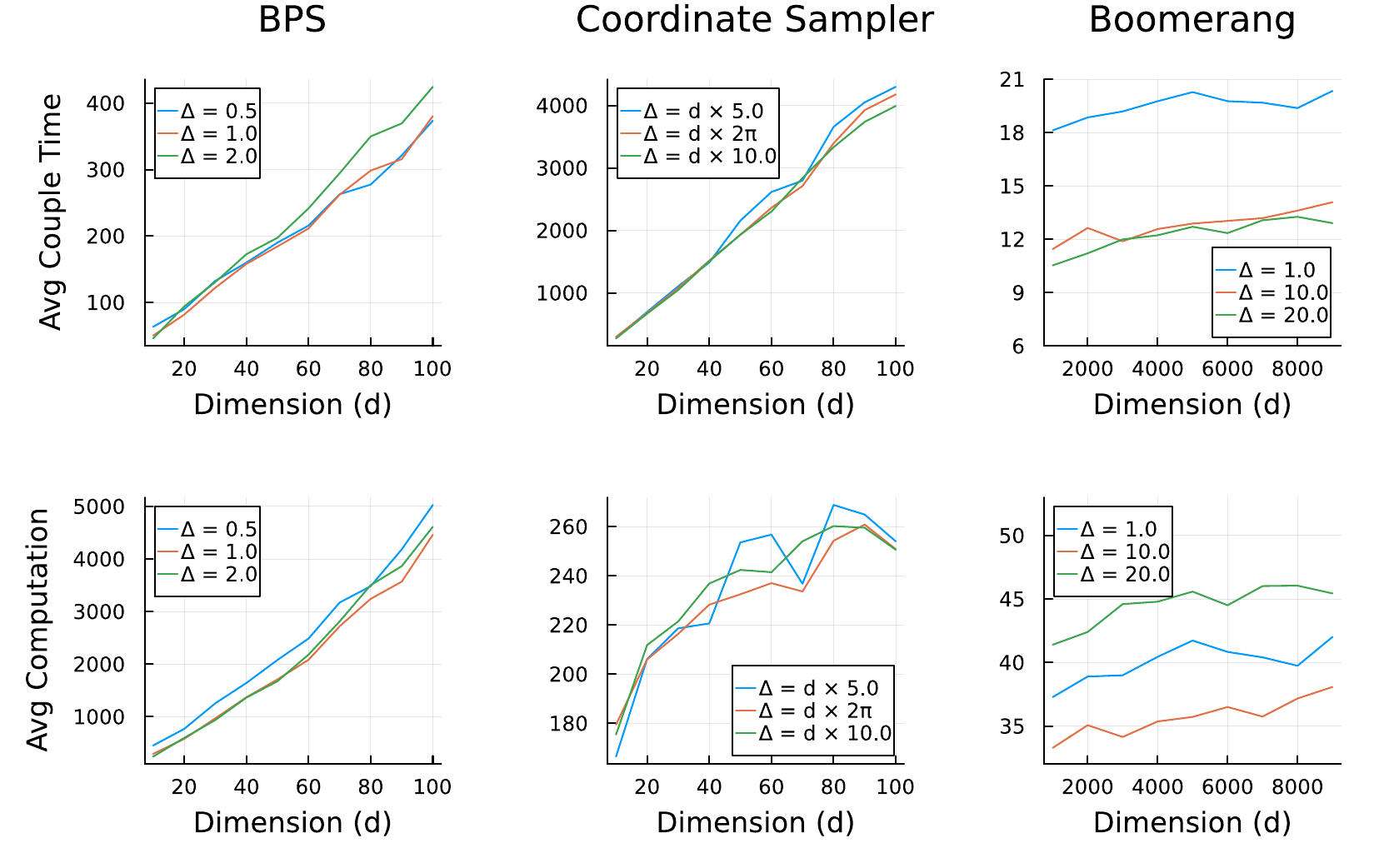}
            \caption{Gaussian distribution scaling example. Average $\Delta$-coupling times of 500 coupled PDMPs for BPS (left), Coordinate sampler (middle) and Boomerang sampler (right). The top row corresponds to the average meeting time $\kappa$ per definition \ref{def:async-coupling} and the bottom row corresponds to the average computational cost.
            \label{fig:GaussScale}}
        \end{figure}
                
        Figure~\ref{fig:GaussScale} suggests that the average $\Delta$-coupling times for BPS and the Coordinate Sampler scale linearly with the dimension. The CS appears to have a better scaling than BPS in terms of computational cost but this comes at the expense of a less efficient sampler. Since CS moves a single coordinate at a time, the efficiency of the sampler is expected to be at least $O(d)$ worse than the BPS or Boomerang. The Boomerang has $\Delta$-coupling times that remain roughly constant with the size of the dimension. This complements similar empirical scaling results from \cite{Heng2019couplings} for coupled HMC.  

        \subsection{Logistic Regression}     
        We consider a Bayesian logistic regression problem as in simulations from \cite{bierkens2020boomerang} to illustrate our method. We simulate a dataset with predictors $a_i \in \mathcal{R}^{16}$ where the first element of $a_i$ is one and the remainder are simulated independently from a normal distribution, and response $y_i \in \{0,1\}$ for $i=1,...,100$. We take a standard normal distribution as the prior for the coefficients. The potential function for our model is,
        $$
        U(x) = \sum_{i=1}^n\left\{\log(1 + e^{a_i^\top x}) - y_ie^{a_i^\top x} \right\} + \frac{1}{2} x^\top x.
        $$
        We consider the task of estimating expectations of the functions $h_i(x) = x_i$ for $i=1,...,16$ and $h_{17}(x) = U(x)$ when starting the sampler from $\pi_0=\mathcal{N}(0_{16},I_{16})$. Small refreshment rates lead to longer deterministic trajectories which give good mixing for marginal means but poor mixing for functions such as $U(x)$ \citep{Deligiannidis2018BPSscale}. The refreshment rate $\lref$ and time-delay $\Delta$ are selected based on running the respective Boomerang and BPS kernels for 1000 units of stochastic time as burnin, followed by 1000 units of stochastic time for estimation. This was repeated 100 times to estimate the average computation (number of gradient evaluations) and the variance of the estimators. We define the inefficiency as the average number of gradient evaluations multiplied by the average of the sum of the variances $\sum_{i=1}^{17}\mathbb{V}[h_i]$. We plot the results in   Figure~\ref{fig:tuning_logistic}.        
        From now on, we set $\Delta=8$  for both samplers, and $\lref=4$ for BPS, $\lambda=3$ for Boomerang.
        \begin{figure}[!htp]
            \centering
            \includegraphics[width=.55\textwidth]{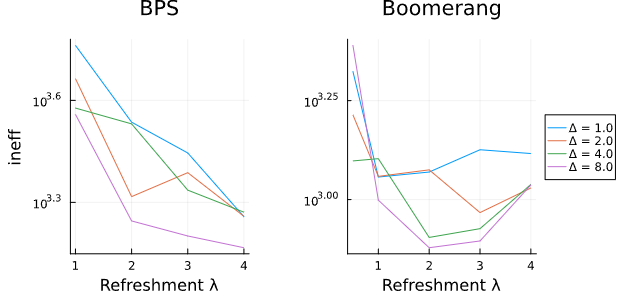}
            \caption{Logistic regression example: tuning parameters. Inefficiency measured as the average number of gradient evaluations multiplied by sum of the variance of estimators for BPS (left) and Boomerang (right).}
            \label{fig:tuning_logistic}
        \end{figure}

        Next we consider the specification of $k$ in the ADDRG estimator. Over the 100 independent runs, the average $\Delta$-meeting times for the BPS and Boomerang were $108.28$ and $31.96$ (respectively) with 95\% quantiles of $291.76$ and $74.14$ (respectively). We apply the guideline of setting $k$ to cover the 95th quantile, so $k=\lceil\frac{291.76}{\Delta}\rceil = 37$ for BPS and $k = \lceil\frac{74.14}{\Delta}\rceil = 10$ for the Boomerang. In parallel, 1000 independent coupled processes were run for both the coupled BPS and Boomerang with these optimal parameters. Table \ref{tab:Relineff} shows the average number of gradient evaluations and sum of the variances of the estimators $ADDRG(k,m,\Delta)$ for multiple values of $m$. 

        \begin{table}[!htp]
        \centering

        \begin{minipage}{0.45\textwidth}
        \begin{tabular}{ccccc}
        \multicolumn{4}{c}{\textbf{Boomerang Sampler}}\\[0.3em]
        $m$   &   Cost   & Variance & Rel. Ineff. \\[0.5em]
        $k$   & $1,361$  & $445.341$& $619.608$ \\
        $50k$ & $41,826$ &  $0.796$ & $37.664$ \\
        $100k$& $83,148$ &  $0.015$ & $1.35$                
        \end{tabular}
        \end{minipage}
                \begin{minipage}{0.45\textwidth}
        \begin{tabular}{ccccc}       
        \multicolumn{4}{c}{\textbf{BPS Sampler}}\\[0.3em]
        $m$   &   Cost   & Variance & Rel. Ineff. \\[0.5em]
        $k$   & $6,874$  & $16.502$& $57.72$ \\
        $10k$ & $51,511$ &  $0.046$ & $1.32$ \\
        $30k$ & $150,514$&  $0.015$ & $1.25$                
        \end{tabular}
        \end{minipage}
                \caption{Logistic regression example: relative inefficiency of the coupled estimators. 
                Cost (average number of gradient evaluations), and variance based on 
                1000 independent runs. Relative inefficiency is relative to the uncoupled version of the sampler.
            \label{tab:Relineff}}
        \end{table}

        Finally, we want to measure the efficiency of our estimators relative to a single process run with the same computation budget and optimal choice of $\lambda$ and $\Delta$. For each value of $m$, the combined amount of process time that the coupled samplers were run for was $\Delta\left[1 + 2(\lceil \frac{\kappa}{\Delta}\rceil -1) + \max\left(1,m-\lceil \frac{\kappa}{\Delta}\rceil\right)\right]$. The average of this time for the 1000 estimators was stored. We ran 1000 further estimators, each using one process and the standard PDMP kernel rather than the coupled kernel. These samplers were afforded 1000 units of burnin time, followed by the average of the combined stochastic time that was used by the coupled samplers. The inefficiency for these estimators was taken as the variance of $\sum_{i=1}^{17}\mathbb{V}[h_i]$ multiplied by the average number of gradient evaluations for the single process. The relative inefficiency was then defined as the ratio of the inefficiency of the estimators from the coupled process to the inefficiency of the estimators using a single process. 

    We add these results to Table \ref{tab:Relineff} (rightmost columns of both sub-tables). For a sufficiently large value of $m$ we find that the estimators attain relative inefficiency close to 1 in this example. For $m=100k$ in Boomerang or $m\geq 10k$ we see that the inefficiency is less than two meaning we can get improved performance relative to the serial versions of the samplers if we have at least two cores. We note the boomerang requires a larger scaling of $m= 100k$ to attain an inefficiency close to 1. This is likely because of the short coupling times where $k=10$ for the Boomerang as opposed to $k=37$ for the BPS.

    \section{Discussion}
        We have presented a general debiasing strategy inspired from~\citet{Jacob2020unbiasedMCMC} for expectations computed under ergodic continuous processes, in particular piecewise deterministic Markov processes. Explicit constructions have been given for the bouncy particle, the boomerang, and the coordinate samplers. We have further shown that, under the hypothesis of a perturbed homogeneous target potential, the bouncy particle sampler returned consistent and efficient estimates. The coordinate sampler required a slight modification, whereby we allowed it to fully switch itself off, to let the two coupled processes recouple in time. This modification did not affect its stationarity nor its ergodicity, but the more precise results obtained in~\citet{wu2019coordinate} do not follow through as easily, and more work would be warranted to obtain drift conditions in our modified version.         
        In the remainder of this section, we discuss some particular points of interest.

        \paragraph{Improving on the $\Delta$-coupling for the coordinate sampler.}
            In Section~\ref{subsec:delta-coupling-CS}, we have made the choice to time-synchronise the two processes when both velocities were null, or to synchronise their positions when the velocities where co-linear, and nothing else otherwise. We should in fact consider additional possible couplings and plan to do so in a future version of this work. Indeed, when only one of the velocities is null, we can still time-synchronise, although the procedure is more complicated given it involves a non-homogeneous Poisson process. Similarly, when both are not null but also not co-linear, we can either still sample from a time-synchronisation coupling, or introduce a position contractive coupling, for example, a coupling on the stochastic times corresponding to an optimal transport coupling for the positions.
        
        \paragraph{The Zig-Zag sampler.}
            At first sight, it may seem that the construction of Section~\ref{sec:coupling-the-coordinate-sampler} should apply to the Zig-Zag sampler \textit{mutatis mutandis}, but it in fact does not. This is because the velocities of the Zig-Zag sampler cannot be killed but only switched. 
            As a consequence, coupling a given coordinate by synchronising the clocks would result in a decoupling phenomenon for all other coordinates, resulting in a never-ending coupling race between these. On the other hand, adding a form of stochasticity in the velocity jumping process would allow for more flexibility in the coupling procedure and perhaps result in a competing algorithm.

        \paragraph{Verifying the assumptions for the boomerang sampler.}
            We have not in this article undertaken the task of verifying that the boomerang sampler verifies the assumptions necessary for the application of our unbiased estimators. However, there is reason to believe that similar properties as for BPS hold for the boomerang sampler. This is because the underlying PDMP structure is highly similar, with the only difference being the deterministic part. While the study of the geometric ergodicity of the boomerang sampler is out of scope for this paper, we conjecture that the analysis of \citet{durmus2020geometric} holds in its entirety. Similarly, other similar PDMPs are likely to emerge in the near future, relying on different types of geometries and/or reference measures for efficiency, which would be likewise covered by our and \citet{durmus2020geometric} work.

        \paragraph{The choice of $\Delta$.}
            In our formulation of the discretised and continuous Rhee \& Glynn estimators, we have assumed that $\Delta > 0$ was deterministic and chosen in advance. This is in fact not a necessary condition, and one could use a random $\Delta$, provided that it can only take positive values.
            In particular, we could select $\Delta$ \emph{after the fact}, by implementing time-coupling events of the form $s^1 + \tau^1 > s^2 + \tau^2$ on the event times $\tau^1$ and $\tau^2$, rather than $s^1 + \tau^1 + \Delta = s^2 + \tau^2$. Doing so would define $\Delta = s^1 + \tau^1 - s^2 - \tau^2$ implicitly (provided that state-coupling is then achieved). 

        \paragraph{Adaptive thinning and local bounds for the Poisson rate.}
             We have presented the construction of our coupling for global bounds of the rate only. In practice, only local bounds may be accessible, and one may need to resort to adaptive thinning schemes. We note that modifying the coupled thinning scheme~\ref{alg:coupled-thinning} to this would be immediate, and requires very little work. Furthermore, when the Poisson processes can be simulated in closed form, it is possible to construct couplings manually rather than via rejection mechanisms (for example using~\ref{alg:coupling-modified-crn}).

    \section*{Acknowlegments}
        Adrien Corenflos acknowledges the financial support provided by UKRI for grant EP/Y014650/1, as part of the ERC Synergy project OCEAN.

    \section*{Individual Contributions}
        The original idea for this article comes from discussions between AC and NC. The methodology was developed primarily by AC in collaboration with NC and MS. The proofs are due to AC and NC. The experimental results are in majority due to MS. Writing was primarily done by AC, after which all authors reviewed the manuscript.

    \bibliographystyle{apalike}
    \bibliography{bib}

    \appendix

    \section{Tuning and coupled MCMC}\label{app:tuning}
        In this section, we extend our discussion of the problems encountered when combining coupling with MCMC algorithms requiring tuning. We present a few possible alternatives and highlight their drawbacks and shortcomings.

        Throughout this section, we consider given a tuning rule for a kernel $K_k(\dd{x}_{k+1} \mid x_k)$ whereby we mean that $K_{k+1}$ is a function of $K_k$, the full history $x_{0:k}$, and the resulting sample $x_{k+1}$. For instance, suppose that $K_k$ is an MRTH kernel with a random-walk proposal $q_k(x_{k+1} \mid x_{k}) = \mathcal{N}(x_{k+1}; x_{k}, \lambda_k I)$, for which we want to tune $\lambda_k$ to achieve an average acceptance rate of 0.234 as in \citet{Roberts2001optimal}, we can compute the running empirical average acceptance rate and decrease $\lambda_k$ if it is found to be too low, or increase it if it is found to be too high.

        Such tuning is done for a number of steps $A > 0$, which is pre-set and chosen by the user, after which the kernel $K$ is fixed to be $K_A$.
        
        \subsection{Independent tuning of coupled chains}
            The first natural idea for coupling tuned kernels is to run the adaptation independently using the same rule, preserving the property of identical marginal distributions for the two chain states $x_k$ and $y_k$. However, doing so results in kernels $K^x$ and $K^y$ which likely do not verify $\norm{K^x(\cdot \mid z) -K^y(\cdot \mid z)}_{\mathrm{TV}} = 0$ for all $z$, so that, when the two chains meet, i.e., for $k$ with $x_k = y_k$, the probability of them separating is non-null. To say otherwise, this procedure does not guarantee that the coupling is faithful, voiding the method of \citet{Jacob2020unbiasedMCMC}.
        \subsection{Tuning then coupling}
            Another solution is to tune the algorithm using a single chain: we obtain a tuned kernel $K$, and can then apply the method of \citet{Jacob2020unbiasedMCMC}. This is well-defined and theoretically grounded. However, it presents a number of disadvantages and questions: 
            \begin{enumerate}
                \item Upfront cost: the tuning needs to be done ahead of time and can't be parallelised, which reduces the interest of the method,
                \item Additional burn-in period: when the tuning procedure is over, we can typically start accumulating statistics from our chain, with no further burn-in. It is however not clear that this is the right thing to do when using \citet{Jacob2020unbiasedMCMC}, given that the coupling may not happen immediately but later on, which contradicts the desire of choosing a burn-in period commensurate to the median coupling time.
                \item Choice of the initial distribution: once the kernel has been fixed, we still need to pick where the two coupled chains need to start. We may start from our initial distribution for $x_0$ again, thereby undoing progress, start both from $x_A$ immediately, reducing potential exploration, or run a few ``decoupling steps'' and obtain $x_{A + C} \sim K^{C}(\cdot \mid x_A)$, $y_{A + C} \sim K^{C}(\cdot \mid x_A)$, approximately independent (provided that $C$ is large enough).
            \end{enumerate}

    \section{On-and-off Coordinate sampler}\label{app:i-o-cs}
        In this section, we discuss the theoretical properties of our I/O-CS modification of the algorithm of~\citet{wu2019coordinate}. We follow their proofs.

        \begin{lemma}
            The I/O-CS is $\pi$-invariant.
        \end{lemma}
        The proof follows the same steps as in~\citet{wu2019coordinate} with very little change besides taking into consideration the case when the velocity is fully switched off.
        \begin{proof}
            The generator of the I/O-CS is given (onto its domain) by 
            \begin{equation}\label{eq:i-o-cs-gen}
                \mathcal{L}f(x, v) = \left\langle \nabla_x f , v \right\rangle + \lambda(x, v)\sum_{i=0}^{2d} \frac{\lambda(x,-v_i)}{\lambda(x)} f(x, v_i) - \lambda(x, v) f(x, v)
            \end{equation}
            where $v_0 = 0$, $v_i = e_i$, $i=1, \ldots, d$, $v_i = -e_i$, $i=d+1, \ldots, 2d$, and $\lambda(x) = \sum_{i=0}^{2d} \lambda(x, v_i)$.

            For any $f$ in the domain of $\mathcal{L}$, we have 
            \begin{align*}
                &\int_x \int_v \mathcal{L}f(x, v) \pi(x, v) \dd{v} \dd{x}\\
                    &=\frac{1}{2d + 1} \sum_{i=0}^{2d} \int_x \mathcal{L}f(x, v_i) \pi(x) \dd{x}\\
                    &= \frac{1}{2d + 1} \sum_{i=0}^{2d} \int_x \left[\left\langle \nabla_x f , v_i \right\rangle + \lambda(x, v_i)\sum_{j=0}^{2d} \frac{\lambda(x,-v_j)}{\lambda(x)} f(x, v_j) - \lambda(x, v_i) f(x, v_i)\right] \pi(x) \dd{x} \\
                    &= \frac{1}{2d + 1} \sum_{i=0}^{2d} \int_x \left\langle \nabla_x f , v_i \right\rangle  \pi(x) \dd{x} +  \frac{1}{2d + 1} \sum_{i=0}^{2d} \sum_{j=0}^{2d} \int_x \lambda(x, v_i) \frac{\lambda(x,-v_j)}{\lambda(x)} f(x, v_j) \pi(x) \dd{x} \\
                    &\quad- \frac{1}{2d + 1} \sum_{i=0}^{2d} \int_x \lambda(x, v_i) f(x, v_i) \pi(x) \dd{x}\\
                    &= \frac{1}{2d + 1} \sum_{i=0}^{2d} \int_x \left\langle -\nabla_x U , v_i \right\rangle f(x, v_i) \pi(x) \dd{x} +  \frac{1}{2d + 1} \sum_{j=0}^{2d} \int_x \lambda(x, -v_j) f(x, v_j) \pi(x) \dd{x} \\
                    &\quad- \frac{1}{2d + 1} \sum_{i=0}^{2d} \int_x \lambda(x, v_i) f(x, v_i) \pi(x) \dd{x}\\
                    &= 0
            \end{align*}
            where the last equality follows from noticing that $\left\langle -\nabla_x U , v_i \right\rangle + \lambda(x, v_i) + \lambda(x, -v_i) = 0$ for all $x$ and $i$.
        \end{proof}

        \begin{lemma}
            The I/O-CS is $\pi$-ergodic as soon as the potential function $U$ is continuously differentiable, and all compact set is petite.
        \end{lemma}
        Again, this directly mirrors Theorem 3 in~\citet{wu2019coordinate}. Following the proof steps is however more tedious and lengthy than proving the invariance only. However, a careful reading of the proof of Lemma 2 in~\citet[Supplementary material]{wu2019coordinate} will highlight that the only change consists in, using their notation for a moment, the value of the quantity~\citep[Supplementary material, page 4]{wu2019coordinate}
        \begin{equation}
            \mathbb{P}\left(\left\{V_k = v_k^*, k=1, \ldots, d+1\right\}\right)
        \end{equation}
        which in our case is equal to $\left(\frac{\lref}{(2d + 1)\lref + 2d K}\right)^{d+1}$ rather than $\left(\frac{\lref}{2d(\lref + K)}\right)^{d+1}$ (note the typo in their proof), only making the value the constant in their lemma slightly smaller, albeit still positive.

    \section{Tractable couplings}\label{app:coupling-tractable}
        In this section, we review the description of tractable maximal couplings, in the sense that we can sample from these in closed form. Given two distributions with densities $p, q$, we have the following~\citet[][Theorem 19.1.6]{Douc2018MC}.
        \begin{theorem}[Maximal coupling representation.]\label{thm:coupling-rep}
            Let $\xi, \xi'$ be two probability distributions, and define 
            \begin{equation}
                \alpha \coloneqq \int \min(\xi, \xi')(\dd{x}) = 1 - \norm{\xi - \xi'}_{\mathrm{TV}},
            \end{equation}
            \begin{equation}
                \nu = \frac{\xi - \min(\xi, \xi')}{1 - \alpha}, \text{ and } \nu' = \frac{\xi' - \min(\xi, \xi')}{1 - \alpha}.
            \end{equation}
            Then a coupling $\gamma$ of $\xi, \xi'$ is maximal if and only if there exists an (arbitrary) coupling $\beta$ of $\nu, \nu'$ such that
            \begin{equation}
                \gamma = (1 - \alpha) \beta(\dd x, \dd y) + \alpha \frac{\min(\xi, \xi')}{\alpha}(\dd x) \delta_x(\dd{y}).
            \end{equation}
        \end{theorem}
        Otherwise said, a coupling is maximal if it is an optimal mixture between $\frac{\min(\xi, \xi')}{\alpha}$ and a coupling on the marginal residuals.
        This means that, for a choice of $\beta$, if we know how to compute $\alpha$ and how to sample from $\min(\xi, \xi')/\alpha$ as well as $\beta$, then we can sample from a maximal coupling of $\xi$ and $\xi'$. 

        \subsection{Coupling shifted exponential distributions}\label{app:coupled-shift}
            For example, when $\xi$ and $\xi'$ are shifted exponential distributions with the same rate, i.e., when $\Xi$ and $\Xi'$ are the corresponding random variables verifying 
            \begin{equation}
                \Xi - \mu \sim \mathrm{Exp}(\lambda), \quad \Xi' - \mu' \sim \mathrm{Exp}(\lambda),
            \end{equation}
            we can compute all the necessary terms explicitly. In this case, and without loss of generality assuming $\mu'=0$ and $\mu > 0$, we have
            \begin{equation}
                \min(\xi, \xi')(x) = \mathbbm{1}_{[\mu, \infty)}(x) \exp(-\lambda x) / \lambda, \quad \alpha = \exp(-\lambda \mu),
            \end{equation}
            so that $\min(\xi, \xi')(x) / \alpha$ has an inverse cumulative density function defined by
            \begin{equation}
                u \mapsto \mu - \frac{\log u}{\lambda},
            \end{equation}
            and the residuals $\nu$ and $\nu'$ have one two, i.e.
            \begin{equation}\label{eq:residual-exponentials}
                v \mapsto -\frac{\phi(v)}{\lambda} \text{ and } w \mapsto -\frac{\psi(w)}{\lambda},
            \end{equation}
            for $\phi(v) \coloneqq \log(\alpha (1-v))$ $\psi(w) \coloneqq \log(1 - (1 - \alpha) w)$,
            respectively. A more general version of this, for exponentials with different rates, can be found in \citet[Supplementary G]{corenflos2022rejection}.

            Coupling the two shifted exponentials then corresponds to a choice of coupling for \eqref{eq:residual-exponentials}, for instance independent, common random number, or antithetic couplings are possible. These respectively correspond to using $v, w \sim \mathcal{U}([0, 1])$, independent, $v=w \sim \mathcal{U}([0, 1])$, or $1 - v = w \sim \mathcal{U}([0, 1])$ in \eqref{eq:residual-exponentials} to simulate from $\beta$.
            
            In Section~\ref{app:coupling-poisson}, we describe the impact of these choices on the resulting behaviour of two $\Delta$-coupled Poisson processes.
            
        \subsection{Illustration}\label{app:coupling-poisson}
            An equivalent formulation of the coupling of Section~\ref{app:coupled-shift} consists in coupling two standard exponentials distributions with the same rate $\lambda$, but trying to maximise $\mathbb{P}(\Xi' + \mu - \Xi  = 0)$ for a given shift variable $\mu > 0$. This is more in line with our definition of the $\Delta$-coupling, and we take this perspective here. In order to understand the different qualitative behaviours of the residual couplings, we consider the following example: $\lambda = 5$ and $\mu = 0.5$. The different choices (independent, common random number, and antithetic) for the residual coupling are illustrated in Figure~\ref{fig:coupling-expong}.
            \begin{figure}[htb]
                \centering
                \includegraphics[width=0.75\textwidth]{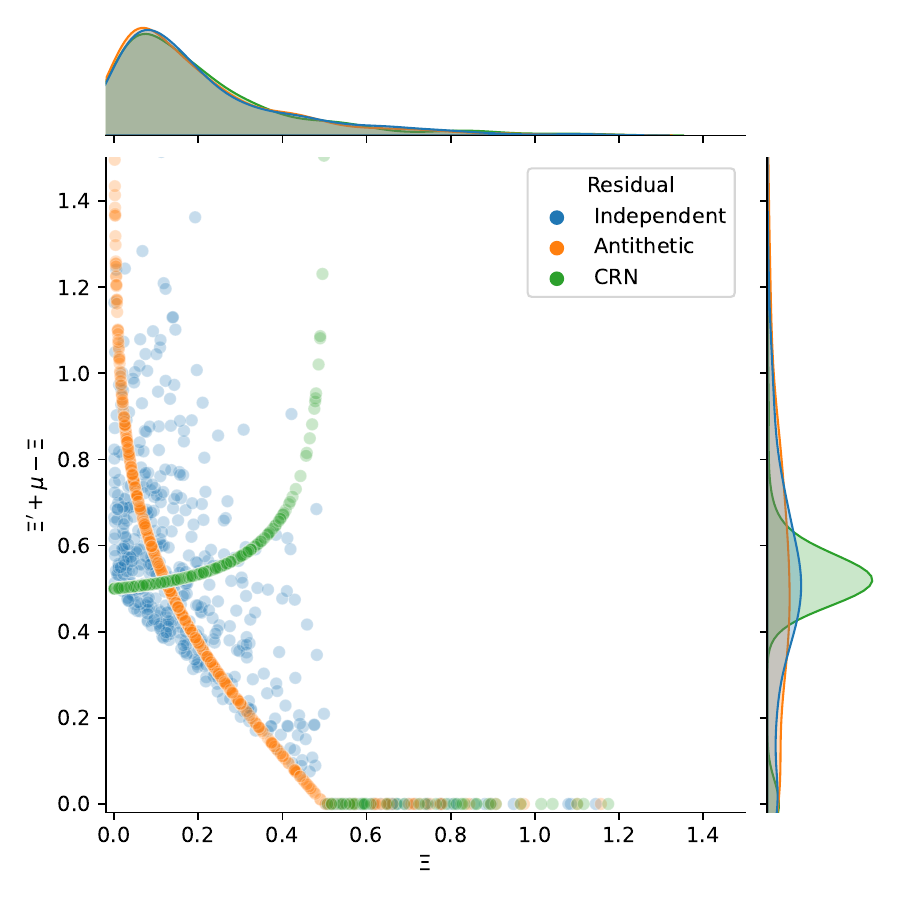}
                \caption{Scatter plot of the coupling for $\Xi' + \mu - \Xi$ and $\Xi$ for different choices of $\beta$ in Theorem~\ref{thm:coupling-rep}.
                The three different choices of residual coupling have the same behaviour ``when the marginals are coupled'', i.e., right of $0.5$ on the x-axis, but result in staunchly different behaviours otherwise.}
                \label{fig:coupling-expong}
            \end{figure}
            As can be seen in the illustration, when using common random numbers for the residuals, $\Xi' + \mu - \Xi$ will either be $0$ (when coupling happens) or greater than $0.5$. In practical terms, this means that Poisson processes coupled using these will not exhibit ``contractive'' behaviour in a geometric sense, and the time distance between the two processes will only ever decrease if coupling happens at all. On the other hand, using antithetic or independent couplings for the residuals results in ``smoother'' transitions between the coupled and uncoupled regimes, resulting in processes that will contract in time, particularly so for the antithetic coupling. 
            This behaviour is similar to the Gaussian case, where reflected residuals in Algorithm~\ref{alg:reflection-maximal} are preferable to independent or common random number ones and is the reason why we implement antithetic residual couplings of exponential distributions whenever possible.

    \section{Proofs}\label{app:proofs}

First, we remark that it is sufficient to establish a similar property
for the coupling of the two discretized Markov chains $\theta_{k}^{i}:=Z_{k\Delta}^{i}$,
$i=1,2$, and their coupling time $\tau:=\lceil\kappa/\Delta\rceil$. 
Indeed, assume that $\mathbb{P}(\tau>k)\leq a\exp(-bk)$,
then 
\[
\mathbb{P}(\kappa>t) \leq\mathbb{P}(\kappa>\Delta\lfloor t/\Delta\rfloor)
  \leq a\exp\left\{ -b\Delta(t/\Delta-1)\right\} 
  \leq a'\exp(-bt)
\]
with $a'=a\exp(b\Delta)$. From now on, we let $M=P_{\Delta}$ (the
Markov kernel associated with these Markov chains), and $\bar{M}$ their
coupling kernel (the distribution of $(\theta_{k}^{1},\theta_{k}^{2})$
given $(\theta_{k-1}^{1},\theta_{k-1}^{2})$). 

Second, this property for coupling of discrete Markov chains is shown
to hold in Proposition 4 (called Proposition 3.4 in the supplement)
of~\citet{Jacob2020unbiasedMCMC}, under the following assumptions: 
\begin{enumerate}
\item Kernel $M$ is $\tilde{\pi}-$invariant, $\varphi-$irreductible and
aperiodic.\label{point-irreducible}

\item There exists a measurable function $V:\mathcal{Z}\rightarrow[1,+\infty)$, such
that 
\begin{equation}
MV(x)\leq\lambda V(x)+b\mathbf{1}_{\mathcal{C}}(x)\label{eq:drift}
\end{equation}
where $\lambda\in(0,1),$ $b>0$, $MV(x):=\int M(x,dy)V(y)$, and
$\mathcal{C}$ is a small set of the form $\mathcal{C}=\left\{ z:V(z)\leq L\right\} $
with $L$ large enough so that $\lambda+b/(1+L)<1$. 
\label{point-drift}
\item There exists $\varepsilon\in(0,1)$ such that 
\[
\inf_{z,z'\in\mathcal{C}}\bar{M}\left((z,z'),\mathcal{D}\right)\geq\varepsilon
\]
where $\mathcal{D}=\left\{ (z,z')\in\mathcal{Z}^{2}:z'=z\right\} $
is the ``diagonal'' of $\mathcal{Z}^{2}$. 
\label{point-coupling}
\end{enumerate}

Regarding point~\ref{point-irreducible}, the $\tilde{\pi}$-invariance of $P_{t}$ (and
therefore $M$) was established by~\citet{durmus2021piecewise}. 
The BPS is also $\varphi-$irreductible and aperiodic as soon as $\lref$ is positive~\citep[Theorem 1]{Bouchard2018BPS}. 

Regarding point~\ref{point-drift}, Corollary 9 in~\citet{durmus2020geometric} implies that (under
appropriate conditions, see comments below) the semigroup $P_{t}$
of the BPS is such that
\begin{equation}
P_{t}V(z)\leq V(z)e^{-A_{1}t}+A_{2}(1-e^{-A_{1}t})\label{eq:durmus_cor9}
\end{equation}
for certain constants $A_{1}$, $A_{2}$, and a certain drift function
$V$ (see below again). Since $L$ may be taken arbitrarily large, take it
so that $L\geq2A_{2}$. Then, taking $\mathcal{C}=\left\{ z:V(z)\leq L\right\} $,
we have that, for $z\notin\mathcal{C}$, $V(z)>L$, and (\ref{eq:durmus_cor9})
implies that $P_{t}V(z)$

\[
P_{t}V(z)\leq\lambda_{t}V(z),\quad\text{with }\lambda_{t}:=e^{-A_{1}t}+(1-e^{-A_{1}t})/2<1,
\]
and for $z\in\mathcal{C}$, we simply have 
\[
P_{t}V(z)\leq\lambda_{t}V(z)+A_{2}
\]
 so (\ref{eq:drift}) holds if we take $t=\Delta$, $b=A_{2}$, and
$\lambda=\lambda_{t}$. 

Note that the drift function $V$ proposed in~\citet{durmus2020geometric} (their
Eq. 22) is fairly complicated, because it is designed to cover simultaneously
alternative sets of assumptions (A1, A2 and either A3, A4, A5, A6
or A7 in their paper). If we focus on the case where the distribution for
refreshing velocities is Gaussian (as described in
Algorithm~\ref{alg:bps-summary-2}), and therefore unbounded, we may focus on
Assumptions A1, A2 and A7, and  
this drift function may be simplified to $V(z)=\exp\left(U^{c}(x)\right)+\exp(\eta\|v\|^{2})$
for a small $\eta>0$, with $z=(x,v)$. What is left to establish
then is that $\mathcal{C}=\left\{ z:V(z)\leq L\right\} $ is indeed
a small set. Lemma 10 in~\citet{durmus2020geometric} states that any compact set
$\mathcal{K}\subset\mathcal{Z}$ is small for the BPS; in particular,
$\mathcal{C}=\{z:V(z)\leq L\}$ is small, since we assume that $U$
is a continuous function which diverges at infinity. 

Regarding point~\ref{point-coupling}, Lemma 14 in~\citet{durmus2020geometric} shows that (under the
same assumptions), for a given compact
set $\mathcal{K}\subset\mathcal{Z}$, there exists $\alpha>0$, and
$t_{0}>0$ such that, for all $t\geq t_{0}$
\[
\|P_{t}(z,\cdot)-P_{t}(z',\cdot)\|_{\mathrm{TV}}\leq2(1-\alpha).
\]
This result is obtained by designing the same coupling construction
that we use (in fact, because we time-synchronise bouncing events too, our coupling probability can only be greater than theirs), and establishing that, under this coupling:
\[
\bar{P}_{t}\left((z,z'),\mathcal{D}\right)\geq\alpha,\quad\forall z,z'\in\mathcal{K}.
\]
We can conclude by remarking again that a compact set is a small set
for the BPS (Lemma 10 in~\citet{durmus2020geometric}, Lemma 2 in~\citet{deligiannidis2019exponential}), and that the specific small
set we consider, $\mathcal{C}=\left\{ z:V(z)\leq L\right\} $ is compact.
Then, we can take $t=\Delta$, large enough so that $\Delta\geq t_{0}=\Delta'$.

    \section{Further recommendations in tuning $\Delta$}\label{app:tuneDelta}

    Here we provide some additional points on tuning the time-delay parameter, $\Delta$, for the coupled process. Recall that the time-delay ensures that the processes are advanced so that they are coupled in time every $\Delta$ increment. We begin by looking at the sensitivity of the average coupling time by running 100 independent simulations from the different coupling kernels on an 8-dimensional standard Normal distribution starting in stationarity. Figure \ref{fig:GaussDelta} shows the results. 
    
    First, note that the time-delay has a purely negative effect on the coupling time for the Boomerang. Since the boomerang sampler is moving according to the Hamiltonian dynamics on the Gaussian it never has a bounce event and so will always be coupled in time. Consequently, $\Delta$ is only a hindering effect for the Boomerang. We therefore recommend using a larger $\Delta$ when the Boomerang reference distribution is close to the target.
            \begin{figure}[!htp]
            \centering
            \includegraphics[width=.6\textwidth]{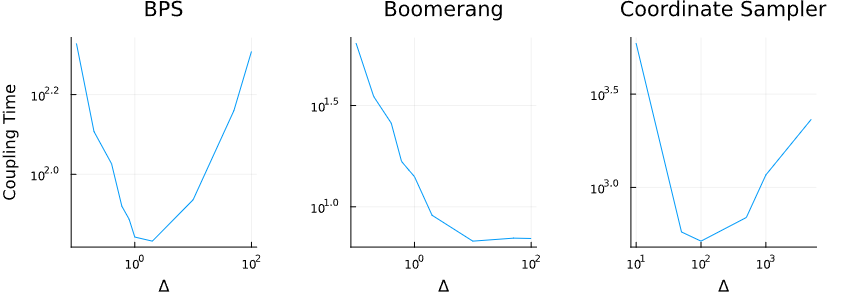}
            \caption{Tuning $\Delta$ for coupling efficiency in the Gaussian distribution scaling example. Plots show the average meeting time $\kappa$ for coupled BPS, coupled Boomerang and the coupled Coordinate Sampler when sampling $N(0_8, I_8)$.}
            \label{fig:GaussDelta}
        \end{figure}
    In comparison, for both the BPS and Coordinate Sampler the process may uncouple in time and a suitable $\Delta$ helps reduce the overall coupling time. For the coordinate sampler we find that scaling $\Delta$ with the dimension, $d$, is necessary to attain good coupling times. Intuitively, the time the process needs to couple $d$ dimensions when one dimension is moved at a time should scale with $d$ so this scaling of $\Delta$ appears sensible. For the BPS sampler we have found that for higher dimensions a smaller $\Delta$ is typically needed. In higher dimensions, the BPS will need to bounce more frequently to move the same distance, and consequently there will be a larger probability that the process will uncouple in time. Similarly, if the process has a poor thinning bound it may be more likely that it will uncouple in time when bouncing, so a smaller choice may be necessary.

\end{document}